\documentclass[a4paper]{cas-sc}

\usepackage{amsthm,amscd}
\usepackage{mathtools}
\usepackage{mathrsfs}

\usepackage{enumitem}
\ExplSyntaxOn
\keys_set:nn { stm / mktitle } { nologo }
\ExplSyntaxOff
\usepackage[numbers,sort&compress]{natbib}

\hypersetup{hidelinks}

\theoremstyle{definition}
\newtheorem{Def}{Definition}[section]

\newtheorem{Rem}[Def]{Remark}

\theoremstyle{plain}
\newtheorem{Prop}[Def]{Proposition}
\newtheorem{Theo}[Def]{Theorem}
\newtheorem{Lem}[Def]{Lemma}
\newtheorem{Coro}[Def]{Corollary}

\DeclareMathOperator{\Hom}{Hom}

\DeclareMathOperator{\Ker}{Ker}

\DeclareMathOperator{\Img}{Im}

\begin{document}

\shorttitle{Covering Radius and Deep Holes of HRS Codes}
\shortauthors{Z. Zhu et al.}

\title[mode=title]{Covering Radius and Deep Holes of Hyperderivative
Reed--Solomon Codes under the NRT Metric}

\author[1]{Zhihao Zhu}[orcid=0009-0008-0595-1960]
\ead{2240501027@cnu.edu.cn}

\author[1]{Haojie Gu}
\ead{2200502051@cnu.edu.cn}

\author[1]{Jun Zhang}[orcid=0000-0003-2173-3302]
\ead{junz@cnu.edu.cn}

\author[2]{Yu Li}
\cormark[1]
\ead{liyu_33@126.com}

\cortext[1]{Corresponding author.}

\affiliation[1]{
    organization={School of Mathematical Sciences,
                  Capital Normal University},
    city={Beijing},
    postcode={100048},
    country={China}
}

\affiliation[2]{
    organization={The 30th Research Institute of China
                  Electronics Technology Group Corporation},
    city={Chengdu},
    postcode={610041},
    country={China}
}

\begin{abstract}
The deep hole problem is a fundamental problem in coding theory,
and it has many important applications in code constructions
and cryptography. The deep hole problem of Reed-Solomon codes
has gained a lot of attention. As a generalization of
Reed-Solomon codes under Hamming metric, we study the covering
radius and deep holes of Hyperderivative Reed-Solomon (HRS)
codes under the Niederreiter-Rosenbloom-Tsfasman (NRT) metric.
We prove that the NRT covering radius of
$\operatorname{HRS}_{k}(\alpha,V)$ is $sr-k$, and characterize
its NRT deep holes via  an MDS-extension criterion and a polynomial-remainder criterion. The latter yields an enumeration formula in terms of the
characteristic polynomial of an associated hyperplane arrangement.
We also construct three families of NRT deep holes and derive exact enumeration formulas in several special cases, including $r=1,2$ and
$k=1,2,sr-1,sr-2,sr-3.$ For full-length evaluation sets and $s\geq 2$, we completely describe the deep holes for $2\leq k\leq sq-1$ whenever $q$ is odd or $k\geq 3$: when $k\geq s$, they are precisely the evaluations of degree-$k$ polynomials, whereas for $k<s$ we give an explicit description.
\end{abstract}

\begin{keywords}
\raggedright
Hyperderivative Reed--Solomon codes
\sep Niederreiter--\allowbreak Rosenbloom--\allowbreak Tsfasman metric
\sep Covering radius
\sep Deep holes
\sep MDS extensions
\sep Hyperplane arrangements

\MSC[2020]{Primary 94B75, 94B05; Secondary 11T71, 94B27}
\end{keywords}

\maketitle

\section{Introduction}

The covering radius of a code is closely related to decoding,
code optimality, and the geometry of metric balls~\cite{CHL+97,CKM+85}.
Its computation is difficult in general. Deciding whether the Hamming
covering radius of a binary linear code is at most a given integer is
$\Pi_{2}^{\mathrm{P}}$-complete~\cite{McL84}.
Exact covering-radius results are known for several important families,
including BCH, cyclic, Reed--Solomon, Reed--Muller, algebraic-geometry,
and MDS codes~\cite{AM76,BGP15,Coh97,DJ91,DS85,Dur94,GPZ60,Hel78,Hel85,
HKM78,Hou93,Hou96,Jan90,MC03,Tie87,Tie90,Wan19,ZZ25}.
In poset metrics, and particularly under the
Niederreiter--Rosenbloom--Tsfasman (NRT) metric~\cite{BGL95,RT97},
the geometry of balls differs substantially from that of Hamming
balls~\cite{CC15}.

A word whose distance from a code equals its covering radius is called
a deep hole. Deep holes have been studied extensively for Reed--Solomon
and generalized Reed--Solomon codes~\cite{CM07,HW16,ZFL13,ZW17,ZCL16}.
Their classification is connected with polynomial interpolation,
finite-field subset-sum and character-sum methods, and MDS
extensions~\cite{Kai17,KW16,WH12}.
Related questions have also been considered for projective,
generalized projective, and twisted Reed--Solomon
codes~\cite{FXZ25,GWZ25,LZS25,XX19,ZW16,ZWK20}, as well as for Reed--Muller
and algebraic-geometry codes~\cite{Oze12,Rot76,ZW23}.
For a survey of the deep-hole problem for generalized Reed--Solomon
codes, see~\cite{ZZ23}.

For full-length Reed--Solomon codes, the Cheng--Murray conjecture
concerns the characterization of deep holes by interpolation
polynomials of degree equal to the code dimension~\cite{CM07}.
Following the partial classifications obtained
in~\cite{ZCL16,Kai17}, recent preprints establish the classification
over sufficiently large fields~\cite{CH26} and settle the corrected
conjecture with the known even-characteristic exceptions~\cite{F26}.
In particular, the dimension-three result in~\cite{F26} is used
in our full-length HRS classification in Section~6.

Hyperderivative Reed--Solomon (HRS) codes are obtained by evaluating
a polynomial and its Hasse derivatives at pairwise distinct field
elements. The ordered derivative blocks fit naturally with the NRT
metric, since the NRT weight of a block is determined by the order
of vanishing at its evaluation point.
HRS codes are MDS with respect to the NRT metric~\cite{Skr01}.
Recent work has addressed their unique decoding, Schur-square
dimensions, and Euclidean duals~\cite{GZ26,GZZ26,LZF+26}.
However, their covering radius and deep holes under the NRT metric
have not been systematically studied.

Our contributions are as follows.
\begin{enumerate}
    \item
    We establish general formulas for the covering radius and deep holes
    of $B$-row lengthenings, together with bounds for row truncations
    (Theorems~\ref{deep hole of extend code} and~\ref{row truncation}).

    \item
    For HRS codes, we determine the NRT covering radius $sr-k$ and
    characterize NRT deep holes by an MDS-extension criterion and
    a polynomial-remainder criterion.
    We then express the number of deep holes through the
    characteristic polynomial of a naturally associated hyperplane
    arrangement
    (Theorems~\ref{NRT covering radius}, \ref{MDS Deep holes},
    \ref{Deep holes} and Corollary~\ref{Count Deep holes}).

    \item
    We construct three explicit families of NRT deep holes and
    derive exact enumeration formulas for several parameter ranges,
    including $r=1,2$ and $k=1,2,sr-1,sr-2,sr-3$
    (Theorem~\ref{When Ev(f) deep hole} and
    Section~\ref{Enumeration of NRT deep holes in special cases}).

    \item
    For the full evaluation set $\mathbb{F}_{q}$, we prove that,
    when $2\leq s\leq k\leq sq-1$ and either $q$ is odd or $k\geq3$,
    the NRT deep holes are precisely the evaluations of
    degree-$k$ polynomials
    (Theorem~\ref{Deep holes of full-length HRS codes}).
    We also obtain an explicit description in the range
    $2\leq k\leq s-1$ under the same condition on $q$ and $k$
    (Proposition~\ref{Conjecture2}).
    Together, these results describe all NRT deep holes in the
    stated full-length parameter range.
\end{enumerate}

The paper is organized as follows.
Section~\ref{Preliminaries} recalls the NRT metric, Hasse derivatives,
Hermite interpolation, and the MDS property of HRS codes, and proves
the row-extension and truncation results.
Section~\ref{NRT covering radius and deep holes} establishes the
covering-radius and deep-hole characterizations and the
hyperplane-arrangement formula.
Section~\ref{Some explicit constructions of NRT deep holes}
gives explicit constructions.
Section~\ref{Enumeration of NRT deep holes in special cases}
derives exact counts in special cases.
Finally, Section~6 proves the full-length classification and the
low-layer description.

\section{Preliminaries}\label{Preliminaries}

In this section, we give some definitions and basic results required in the paper.

\medskip

Let $\mathbb{N}=\{0,1,2,\dots\}.$ For $k\in\mathbb{N}$, we denote by $\mathbb{F}_{q}[x]_{\leq k}$ the set of polynomials over $\mathbb{F}_{q}$ of degree at most $k$ and by $\mathbb{F}_{q}[x]_{k}$ the set of polynomials over $\mathbb{F}_{q}$ of degree exactly $k$. For integers $a\leq b$, we write $[a,b]=\{a,a+1,\dots,b\}$ and set $[n]:=[1,n]$. Given an $r$-tuple $\alpha=(\alpha_{j})_{j\in[r]}$ of pairwise distinct elements of $\mathbb{F}_q$, we define $\pi_{\alpha}(x)=\prod_{j=1}^{r}(x-\alpha_{j})$ and $\eta_{j}=\pi'_{\alpha}(\alpha_{j})=\prod_{\ell\neq j}(\alpha_{j}-\alpha_{\ell})$ for $j\in[r]$. For $A\in \mathbb{F}_{q}^{s\times r}$ and $i\in[s]$,
let $A^{(i)}
\in\mathbb{F}_{q}^{i\times r}$
denote the submatrix consisting of the first $i$ rows of $A$.
For any subset $\mathscr S\subseteq\mathbb{F}_{q}^{s\times r}$,
define $\mathscr S^{(i)}
=\{A^{(i)}\mid A\in\mathscr S\}.$
In particular, $\mathscr C^{(i)}$ is called the $i$-row truncation
of a code $\mathscr C\subseteq\mathbb{F}_{q}^{s\times r}$.

\medskip

We first recall the definition of the NRT weight~\cite[p. 55]{RT97}, and then introduce the NRT covering radius and NRT deep holes.

Let $\mathbb{F}_{q}$ be a finite field of size $q$ and characteristic $p$. Given $\mathbf{x}=(x_{1},\dots,x_{s})\in\mathbb{F}_{q}^{s},$ the NRT-weight of $\mathbf{x}$ is defined to be  $$wt_{\operatorname{RT}}(\mathbf{x})=\begin{cases}
    s-\iota(\mathbf{x})+1 & \text{if}\quad \mathbf{x} \neq \mathbf{0}; \\
    0  & \text{if}  \quad \mathbf{x}=\mathbf{0}.
\end{cases}$$
where $\iota(\mathbf{x})\coloneqq \min\{i\in[s]\mid x_{i}\neq 0\}.$ The NRT-weight of $A:= (a_{ij})_{\substack{i=1,\dots,s \\ j=1,\dots,r}} \in \mathbb{F}_{q}^{s\times r}$ is defined to be $wt_{\operatorname{RT}}(A)=\sum_{j=1}^{r}wt_{\operatorname{RT}}(A_{j}),$
where $A_{j}\coloneqq (a_{1j},\dots,a_{sj})$.  The NRT-distance between two matrices $A,B\in \mathbb{F}_{q}^{s\times r}$ is defined by $d_{\operatorname{RT}}(A,B)=wt_{\operatorname{RT}}(A-B).$ Let $\mathscr{C}\subseteq\mathbb{F}_{q}^{s\times r}$ be a code. The minimum NRT distance of $\mathscr{C}$ is $d_{\operatorname{RT}}(\mathscr{C})=\min_{\substack{
A,B \in \mathscr{C} \\ A\neq B}}{d_{\operatorname{RT}}(A,B)}.$ The error NRT distance of any word $U\in \mathbb{F}_{q}^{s\times r}$ to $\mathscr{C}$ is defined to be $d_{\operatorname{RT}}(U,\mathscr{C})=\min\{d_{\operatorname{RT}}(U,C)\mid C\in\mathscr{C}\}.$ The maximum NRT error distance $\rho_{\operatorname{RT}}(\mathscr{C})=\max\{d_{\operatorname{RT}}(U,\mathscr{C})\mid U\in \mathbb{F}_{q}^{s\times r}\}$
is called the NRT covering radius of $\mathscr{C}$, and the words achieving maximum NRT error distance are called NRT deep holes of $\mathscr{C}$. We denote the set of all NRT deep holes of $\mathscr{C}$ by $\operatorname{DH}(\mathscr{C}).$

\medskip

Next, we introduce the Hyperderivative Reed-Solomon codes and express the NRT weight of $\operatorname{Ev}_{\alpha,V}(f)$ in terms of the orders of vanishing of $f$ at the evaluation points.

\begin{Def}
Let
$
f(x)=\displaystyle\sum_{i=0}^{n} a_i x^i \in \mathbb{F}_q[x]
$
and $j\in\mathbb{N}$. The $j$-th hyperderivative, or Hasse
derivative, of $f(x)$ is defined by
$$    \partial^{(j)}f(x)
=
\sum_{i=0}^{n}\binom{i}{j}a_i x^{i-j},$$
where
$$\binom{i}{j}:=\begin{cases}\displaystyle \frac{i!}{(i-j)!j!}, &\text{if}\quad 0\leq j\leq i;\\
0, &\text{otherwise}.
\end{cases}
$$
Here the coefficient $\binom{i}{j}$ is understood as the reduction modulo
$p$ of the corresponding integer binomial coefficient.
\end{Def}
For any $u\in \mathbb{F}_q$, one has the Hasse-Taylor expansion of $f(x)$ at $u$: $f(x)=\sum_{j\ge 0}\partial^{(j)}f(u)(x-u)^j.$

\begin{Def}
    Let $f(x)\in \mathbb{F}_{q}[x]$ and $u\in \mathbb{F}_{q}$. The order of vanishing of $f(x)$ at $u$ is defined as $$\nu_{f}(u)=\max\{i\in\mathbb{N}: (x-u)^{i}\mid f(x) \}.$$
    We use the convention that \(\nu_0(u)=\infty\).
\end{Def}

\begin{Def} Let $\alpha=(\alpha_{1},\dots,\alpha_{r})$ be an $r$-tuple of pairwise distinct elements of $\mathbb{F}_q.$ Let $V= (v_{ij})_{\substack{i=1,\dots,s \\ j=1,\dots,r}} \in (\mathbb{F}_{q}^{\star})^{s\times r}$. Let $1\leq k \leq sr$.
The Hyperderivative Reed-Solomon code $\operatorname{HRS}_{k}(\alpha,V)$ is defined to be
\[
\operatorname{HRS}_k(\alpha,V)
=
\left\{
 \operatorname{Ev}_{\alpha,V}(f)
:
f\in\mathbb{F}_{q}[x]_{\leq k-1}
\right\},
\]
where the evaluation map
\begin{align*}
    \operatorname{Ev}_{\alpha,V}:\mathbb{F}_{q}[x]
&\longrightarrow
\mathbb{F}_{q}^{s\times r} \\
f(x)&\longmapsto
\begin{bmatrix}
v_{1,1}\partial^{(0)}f(\alpha_1) &
v_{1,2}\partial^{(0)}f(\alpha_2) &
\cdots &
v_{1,r}\partial^{(0)}f(\alpha_r)
\\
v_{2,1}\partial^{(1)}f(\alpha_1) &
v_{2,2}\partial^{(1)}f(\alpha_2) &
\cdots &
v_{2,r}\partial^{(1)}f(\alpha_r)
\\
\vdots & \vdots &  & \vdots
\\
v_{s,1}\partial^{(s-1)}f(\alpha_1) &
v_{s,2}\partial^{(s-1)}f(\alpha_2) &
\cdots &
v_{s,r}\partial^{(s-1)}f(\alpha_r)
\end{bmatrix}.
\end{align*}
\end{Def}

\medskip

The following results follow directly from the definition of Hasse derivatives:
\begin{enumerate}
    \item For all $m\in\mathbb{N}$, we have $$\nu_{f}(u)=m\Longleftrightarrow
    \begin{cases}
        \partial^{(j)}f(u)=0, \quad \forall\ 0\leq j\leq m-1;\\
        \partial^{(m)}f(u)\neq 0.
    \end{cases}$$
    \item For any integer $k>\deg(f(x))$, we have $\partial^{(k)}f(x)=0.$
\end{enumerate}
Consequently, for any $f(x)\in\mathbb{F}_{q}[x]$,  we have
$$wt_{\operatorname{RT}}(v_{1j}\partial^{(0)}f(\alpha_{j}),v_{2j}\partial^{(1)}f(\alpha_{j}),\dots,v_{sj}\partial^{(s-1)}f(\alpha_{j}))=s-\min\{\nu_{f}(\alpha_{j}),s\}$$
and $$wt_{\operatorname{RT}}(\operatorname{Ev}_{\alpha,V}(f))=sr-\sum_{j=1}^{r}\min\{\nu_{f}(\alpha_{j}),s\}.$$

By Hermite interpolation, we can deduce the following lemma which is useful when dealing with NRT deep holes.
\begin{Lem} \cite{Her78}\label{iso}
     Let $\alpha=(\alpha_{1},\dots,\alpha_{r})$ be an $r$-tuple of pairwise distinct elements of $\mathbb{F}_q.$ Let $V= (v_{ij})_{\substack{i=1,\dots,s \\ j=1,\dots,r}} \in (\mathbb{F}_{q}^{\star})^{s\times r}$. Then the restricted evaluation map $ \operatorname{Ev}_{\alpha,V}|_{\leq sr-1}:\mathbb{F}_{q}[x]_{\leq sr-1}\longrightarrow
\mathbb{F}_{q}^{s\times r} $ 
is an $\mathbb{F}_{q}$-linear isomorphism.
\end{Lem}

For every $U\in\mathbb F_q^{s\times r}$, we denote by
$f_U(x)=(\operatorname{Ev}_{\alpha,V}|_{\leq sr-1})^{-1}(U)\in\mathbb{F}_{q}[x]_{\leq sr-1}$ the unique Hermite interpolation polynomial representing $U.$  By the Hermite interpolation formula, we obtain the following explicit expression for $f_{U}$:
\begin{Prop}\label{adic expansion}
    Let $U=(u_{i,j})_{\substack{i=1,\dots,s \\ j=1,\dots,r}}\in\mathbb F_q^{s\times r}$. Define  the Lagrange basis polynomial $L_{j}(x)=\frac{\pi_{\alpha}(x)}{(x-\alpha_{j})\eta_{j}}$ for $j\in[r]$. Then $f_{U}(x)$ admits a unique $\pi_{\alpha}$-adic expansion:
$$f_{U}(x)=\sum_{i=0}^{s-1}f_{i}(x)\pi_{\alpha}(x)^{i}$$
where $\displaystyle f_{i}(x)\in\mathbb{F}_{q}[x]_{<r}$ is the Lagrange interpolation polynomial satisfying $$f_{i}(\alpha_{j})=\eta_{j}^{-i}\left(\frac{u_{i+1,j}}{v_{i+1,j}}-\sum_{m=0}^{i-1}\partial^{(i)}(f_{m}\pi_{\alpha}^{m})(\alpha_{j})\right)\quad\text{for all } j\in[r]$$
that is, $$\displaystyle f_{i}(x)=\sum_{j=1}^{r}\eta_{j}^{-i}\left(\frac{u_{i+1,j}}{v_{i+1,j}}-\sum_{m=0}^{i-1}\partial^{(i)}(f_{m}\pi_{\alpha}^{m})(\alpha_{j})\right)L_{j}(x).$$
\end{Prop}

The following results show that HRS codes are MDS codes with respect to the NRT metric, which is one of the main reasons for considering this family of codes.

\begin{Prop}{\rm \cite[Theorem 1]{RT97}}
    Let $\mathscr{C}\subseteq \mathbb{F}_{q}^{s\times r}$ be a code with NRT distance $d_{\operatorname{RT}}(\mathscr{C})$. Then $$\log_{q}(\left|\mathscr{C}\right|)\leq sr-d_{\operatorname{RT}}(\mathscr{C})+1.$$
\end{Prop}

\begin{Def}{\rm \cite[Definition 2.3]{HK08}}
    A code $\mathscr{C}\subseteq\mathbb{F}_{q}^{s\times r}$ with $\left|\mathscr{C}\right|=q^{k}$ and minimum NRT distance $d_{\operatorname{RT}}(\mathscr{C})$ satisfying $k=sr-d_{\operatorname{RT}}(\mathscr{C})+1$ is said to be a Maximum Distance Separable NRT-code (or simply MDS NRT code).
\end{Def}

\begin{Theo}{\rm \cite[Theorem 5.1]{Skr01}}
     Let $\alpha=(\alpha_{1},\dots,\alpha_{r})$ be an $r$-tuple of pairwise distinct elements of $\mathbb{F}_q.$ Let $V= (v_{ij})_{\substack{i=1,\dots,s \\ j=1,\dots,r}} \in (\mathbb{F}_{q}^{\star})^{s\times r}$. Let $1\leq k \leq sr$. Then $$d_{\operatorname{RT}}(\operatorname{HRS}_{k}(\alpha,V))=sr-k+1$$
     i.e., $\operatorname{HRS}_{k}(\alpha,V)$ is an MDS NRT code.
\end{Theo}

Motivated by the definition of lengthening codes in $\mathbb{F}_{q}^{n}$, we introduce the following definition of $B$-row lengthening codes and investigate the relationship of NRT covering radius and NRT deep holes between $\mathscr
C\subseteq\mathbb{F}_{q}^{s\times r}$ and its $B$-row lengthening code $\binom{\mathscr{C}}{B}$. In particular, if $s\geq k$, then $\operatorname{HRS}_{k}(\alpha,V)$ can be viewed as the $\mathbf{0}^{(s-k)\times r}$-row lengthening code of $\operatorname{HRS}_{k}(\alpha,V^{(k)})$ where $V^{(k)}$ is the submatrix of $V$ consisting of its first $k$ rows. The corresponding result will be established in Proposition~\ref{Conjecture2}.
\begin{Def}
  Let $\mathscr{C}\subseteq \mathbb{F}_{q}^{s\times r}$ be a code. For $t\geq 1$ and $B\in\mathbb{F}_{q}^{t\times r}$, the $B$-row lengthening of $\mathscr{C}$ is defined as $$\binom{\mathscr{C}}{B}=\left\{\begin{pmatrix}
        C\\
        B
    \end{pmatrix}\ \middle| \ C\in\mathscr{C}\right\}\subseteq\mathbb{F}_{q}^{(s+t)\times r}.$$
\end{Def}

\begin{Theo}\label{deep hole of extend code}
    Let $\mathscr{C}\subseteq \mathbb{F}_{q}^{s\times r}$ be a code, $t\geq 1$ and $B\in\mathbb{F}_{q}^{t\times r}$. For  $X\in\mathbb{F}_{q}^{s\times r}$, $Z\in\mathbb{F}_{q}^{t\times r}$ and $C\in\mathscr{C}$, define $\Delta_{t,B}(X,Z;C)=\displaystyle\sum_{\substack{j\in[r]\\X_{j}=C_{j}}}(t-wt_{\operatorname{RT}}(Z_{j}-B_{j}))$. Then $$\rho_{\operatorname{RT}}\left(\binom{\mathscr{C}}{B}\right)=\rho_{\operatorname{RT}}(\mathscr{C})+tr$$
    and $$\operatorname{DH}\left(\binom{\mathscr{C}}{B}\right)=\left\{\begin{pmatrix}
        X\\
        Z
    \end{pmatrix}\ \middle| \ \begin{array}{c}
         X\in \operatorname{DH}(\mathscr{C}),Z\in\mathbb{F}_{q}^{t\times r} \ \text{such that}\\
         \Delta_{t,B}(X,Z;C)\leq d_{\operatorname{RT}}(X,C)-\rho_{\operatorname{RT}}(\mathscr{C})\ \text{for all} \ C\in\mathscr{C}
    \end{array}\right\}.$$
\end{Theo}

{\it Proof.} From the construction of $\Delta_{t,B}(X,Z;C)$ we have 
$$d_{\operatorname{RT}}\left(\begin{pmatrix}
        X\\
        Z
\end{pmatrix},\begin{pmatrix}
        C\\
       B\end{pmatrix}\right)=d_{\operatorname{RT}}(X,C)+tr-\Delta_{t,B}(X,Z;C).$$
Consequently, \begin{align*}
        d_{\operatorname{RT}}\left(\begin{pmatrix}
        X\\
        Z
    \end{pmatrix},\binom{\mathscr{C}}{B}\right)=\min_{C\in\mathscr{C}}\left(d_{\operatorname{RT}}(X,C)-\Delta_{t,B}(X,Z;C)\right)+tr
    \leq \min_{C\in\mathscr{C}}\left(d_{\operatorname{RT}}(X,C)\right)+tr\leq\rho_{\operatorname{RT}}(\mathscr{C})+tr.
    \end{align*} This implies $\rho_{\operatorname{RT}}\left(\begin{pmatrix}
        \mathscr{C}\\
        B
    \end{pmatrix}\right)\leq\rho_{\operatorname{RT}}(\mathscr{C})+tr$. Conversely, if we take $X_{0}\in\operatorname{DH}(\mathscr{C})$ and $Z_{0}\in\mathbb{F}_{q}^{t\times r}$ such that the first row of $Z_{0}-B$ lies in $(\mathbb{F}_{q}^{*})^{1\times r}$, then $\Delta_{t,B}(X_{0},Z_{0};C)=0$ for all  $C\in\mathscr{C}$. In this case, $d_{\operatorname{RT}}\left(\begin{pmatrix}
        X_{0}\\
        Z_{0}
    \end{pmatrix},\binom{\mathscr{C}}{B}\right)=\min_{C\in\mathscr{C}}\left(d_{\operatorname{RT}}(X_{0},C)\right)+tr=d_{\operatorname{RT}}(X_{0},\mathscr{C})+tr=\rho_{\operatorname{RT}}(\mathscr{C})+tr.$ Thus $$\rho_{\operatorname{RT}}\left(\binom{\mathscr{C}}{B}\right)=\rho_{\operatorname{RT}}(\mathscr{C})+tr.$$
It follows from the above proof that $\begin{pmatrix}
        X\\
        Z
    \end{pmatrix}\in \operatorname{DH}\left(\begin{pmatrix}
        \mathscr{C}\\
        B
    \end{pmatrix}\right)$ if and only if $$\min_{C\in\mathscr{C}}\left(d_{\operatorname{RT}}(X,C)-\Delta_{t,B}(X,Z;C)\right)=\rho_{\operatorname{RT}}(\mathscr{C}).$$
Suppose first that this equality holds. Since $\Delta_{t,B}(X,Z;C)\geq 0$ for all $C\in\mathscr{C}$, $$\rho_{\operatorname{RT}}(\mathscr{C})=\min_{C\in\mathscr{C}}\left(d_{\operatorname{RT}}(X,C)-\Delta_{t,B}(X,Z;C)\right)\leq\min_{C\in\mathscr{C}}\left(d_{\operatorname{RT}}(X,C)\right)=d_{\operatorname{RT}}(X,\mathscr{C})\leq \rho_{\operatorname{RT}}(\mathscr{C}).$$
    This implies $d_{\operatorname{RT}}(X,\mathscr{C})=\rho_{\operatorname{RT}}(\mathscr{C}),$ that is, $X\in\operatorname{DH}(\mathscr{C}).$ Furthermore, the equality $\displaystyle\min_{C\in\mathscr{C}}\left(d_{\operatorname{RT}}(X,C)-\Delta_{t,B}(X,Z;C)\right)=\rho_{\operatorname{RT}}(\mathscr{C})$ implies that, for every $C\in\mathscr{C}$, $d_{\operatorname{RT}}(X,C)-\Delta_{t,B}(X,Z;C)\geq\rho_{\operatorname{RT}}(\mathscr{C})$, which is equivalent to $\Delta_{t,B}(X,Z;C)\leq d_{\operatorname{RT}}(X,C)-\rho_{\operatorname{RT}}(\mathscr{C}).$ Conversely, suppose that $X\in\operatorname{DH}(\mathscr{C})$ and $\Delta_{t,B}(X,Z;C)\leq d_{\operatorname{RT}}(X,C)-\rho_{\operatorname{RT}}(\mathscr{C})$ for every $C\in\mathscr{C}.$ Then $$d_{\operatorname{RT}}\left(\begin{pmatrix}
        X\\
        Z
    \end{pmatrix},\binom{\mathscr{C}}{B}\right)=\min_{C\in\mathscr{C}}\left(d_{\operatorname{RT}}(X,C)-\Delta_{t,B}(X,Z;C)\right)+tr\geq \rho_{\operatorname{RT}}(\mathscr{C})+tr.$$
    This implies $d_{\operatorname{RT}}\left(\begin{pmatrix}
        X\\
        Z
    \end{pmatrix},\begin{pmatrix}
        \mathscr{C}\\
        B
    \end{pmatrix}\right)=\rho_{\operatorname{RT}}(\mathscr{C})+tr=\rho_{\operatorname{RT}}\left(\begin{pmatrix}
        \mathscr{C}\\
        B
    \end{pmatrix}\right)$, that is, $\begin{pmatrix}
        X\\
        Z\end{pmatrix}\in\operatorname{DH}\left(\begin{pmatrix}
        \mathscr{C}\\
        B
    \end{pmatrix}\right).$
\qed

When $B=\mathbf{0}_{t\times r}$, we simply write $\Delta_{t}(X,Z;C)=\Delta_{t,\mathbf{0}_{t\times r}}(X,Z;C)$.

\medskip

Similarly, we consider the $i$-row truncation $\mathscr C^{(i)}$ of $\mathscr C$, obtained by puncturing the last $s-i$ rows of each codeword and investigate how its covering radius and deep holes relate to those of $\mathscr C$.

\begin{Theo}\label{row truncation}
    Let $\mathscr{C}\subseteq\mathbb{F}_{q}^{s\times r}$ be a code and $i\in[s]$. Then \[
\rho_{\operatorname{RT}}(\mathscr C^{(i)})
\le\rho_{\operatorname{RT}}(\mathscr C)
\le\rho_{\operatorname{RT}}(\mathscr C^{(i)})+(s-i)r.
\]Moreover, if $
\rho_{\operatorname{RT}}(\mathscr C)
=
\rho_{\operatorname{RT}}(\mathscr C^{(i)})+(s-i)r
$ then $$(\operatorname{DH}(\mathscr{C}))^{(i)}\subseteq\operatorname{DH}(\mathscr{C}^{(i)}).$$
\end{Theo}

{\it Proof.} The case $i=s$ is immediate. Suppose $i<s$.
For any
$
U=\begin{pmatrix}X\\Z\end{pmatrix}
\in\mathbb F_q^{s\times r}
\quad\text{and}\quad
\begin{pmatrix}A\\B\end{pmatrix}\in\mathscr C,
$
where $X,A\in\mathbb F_q^{i\times r}$ and
$Z,B\in\mathbb F_q^{(s-i)\times r}$, we have
\begin{align*}
d_{\operatorname{RT}}(X,\mathscr C^{(i)})
&=
\min_{\binom{A}{B}\in\mathscr C}
d_{\operatorname{RT}}(X,A)\le
\min_{\binom{A}{B}\in\mathscr C}
\left(
d_{\operatorname{RT}}(X,A)+(s-i)r
-\Delta_{s-i,B}(X,Z;A)
\right)\\
&=
d_{\operatorname{RT}}(U,\mathscr C)\le
\min_{\binom{A}{B}\in\mathscr C}
d_{\operatorname{RT}}(X,A)+(s-i)r=
d_{\operatorname{RT}}(X,\mathscr C^{(i)})+(s-i)r.
\end{align*}
Take $U_{0}=\begin{pmatrix}
        X_{0}\\
        Z_{0}\end{pmatrix}\in\operatorname{DH}(\mathscr{C}).$ Then
$$
\rho_{\operatorname{RT}}(\mathscr C)=d_{\operatorname{RT}}(U_{0},\mathscr{C})
\leq d_{\operatorname{RT}}(X_{0},\mathscr C^{(i)})+(s-i)r\leq 
\rho_{\operatorname{RT}}(\mathscr C^{(i)})+(s-i)r.
$$
On the other hand, take
$X_0\in\operatorname{DH}(\mathscr C^{(i)})$
and $Z_0=\mathbf{0}$.
Then
$
\rho_{\operatorname{RT}}(\mathscr C^{(i)})=d_{\operatorname{RT}}(X_0,\mathscr C^{(i)})\le d_{\operatorname{RT}}\left(
\begin{pmatrix}X_0\\Z_0\end{pmatrix},\mathscr C
\right)\leq \rho_{\operatorname{RT}}(\mathscr C) .
$
Thus
$$
\rho_{\operatorname{RT}}(\mathscr C^{(i)})
\le
\rho_{\operatorname{RT}}(\mathscr C)
\le
\rho_{\operatorname{RT}}(\mathscr C^{(i)})+(s-i)r.
$$

Now suppose that
$
\rho_{\operatorname{RT}}(\mathscr C)
=
\rho_{\operatorname{RT}}(\mathscr C^{(i)})+(s-i)r,
$
and take
$U=\begin{pmatrix}X\\Z\end{pmatrix}
\in\operatorname{DH}(\mathscr C)$.
Since $\Delta_{s-i,B}(X,Z;A)\ge0$, we obtain
\begin{align*}
\rho_{\operatorname{RT}}(\mathscr C^{(i)})
=
\rho_{\operatorname{RT}}(\mathscr C)-(s-i)r=
d_{\operatorname{RT}}(U,\mathscr C)-(s-i)r\leq 
d_{\operatorname{RT}}(X,\mathscr C^{(i)})\le
\rho_{\operatorname{RT}}(\mathscr C^{(i)}).
\end{align*}
Hence
$d_{\operatorname{RT}}(X,\mathscr C^{(i)})
=\rho_{\operatorname{RT}}(\mathscr C^{(i)})$,
so $U^{(i)}=X\in\operatorname{DH}(\mathscr C^{(i)})$.
Therefore,
$
(\operatorname{DH}(\mathscr C))^{(i)}
\subseteq
\operatorname{DH}(\mathscr C^{(i)}).
$
\qed

\section{NRT covering radius and deep holes of $\operatorname{HRS}_{k}(\alpha,V)$}\label{NRT covering radius and deep holes}
In this section, we shall give the NRT covering radius of HRS codes, develop two equivalent characterizations of NRT deep holes and derive a general formula for the number of NRT deep holes of HRS codes using the characteristic polynomial of a hyperplane arrangement.

\medskip

Throughout the paper, we shall fix an $r$-tuple $\alpha=(\alpha_{1},\dots,\alpha_{r})$ of pairwise distinct elements of $\mathbb{F}_{q}$, a matrix $V= (v_{ij})_{\substack{i=1,\dots,s \\ j=1,\dots,r}} \in (\mathbb{F}_{q}^{\star})^{s\times r}$ and an integer $1\leq k \leq sr$.
\begin{Theo}\label{NRT covering radius}
    $\rho_{\operatorname{RT}}(\operatorname{HRS}_{k}(\alpha,V))=sr-k.$
\end{Theo}

{\it Proof.} 
Since $1\le k\le sr$, there are integers $m_1,\dots,m_r$ satisfying $m_j\in[0,s]$ and $\sum_{j=1}^{r}m_j=k.$ Set $M(x)=\prod_{j=1}^{r}(x-\alpha_j)^{m_j}.$ For any $f_{U}(x)\in\mathbb{F}_{q}[x]_{\leq sr-1}$, there exists $h(x)\in\mathbb{F}_{q}[x]_{\leq k-1}$ such that $h(x)\equiv f_U(x)\bmod{M(x)}.$ Note that $\operatorname{Ev}_{\alpha,V}(h)\in \operatorname{HRS}_{k}(\alpha,V)$ and $(x-\alpha_j)^{m_j}\mid f_U(x)-h(x)$ for $j\in[r].$ Hence $\nu_{f_U-h}(\alpha_j)\ge m_j.$
    Since $m_j\le s$, we have $\min\{\nu_{f_U-h}(\alpha_j),s\}\ge m_j.$
    Therefore,
    \[
    \begin{aligned}
    d_{\operatorname{RT}}(U,\operatorname{HRS}_{k}(\alpha,V))
    &\le
    d_{\operatorname{RT}}\bigl(U,\operatorname{Ev}_{\alpha,V}(h)\bigr)=
    wt_{\operatorname{RT}}\bigl(\operatorname{Ev}_{\alpha,V}(f_U-h)\bigr)\\
    &=
    sr-
    \sum_{j=1}^{r}
    \min\{\nu_{f_U-h}(\alpha_j),s\}\leq
    sr-
    \sum_{j=1}^{r}
    m_{j}= sr-k.
    \end{aligned}
    \]
    It follows that $\rho_{\operatorname{RT}}(\operatorname{HRS}_{k}(\alpha,V))\le sr-k.$

   Set $f_0(x)=x^k$ and $U_0=\operatorname{Ev}_{\alpha,V}(f_0)$.
     For any \(h\in\mathbb{F}_{q}[x]_{\leq k-1}\), $f_0(x)-h(x)=x^k-h(x)\in\mathbb{F}_{q}[x]_{k}.$ This implies 
 $\sum_{j=1}^{r}
    \nu_{f_0-h}(\alpha_j)
    \le k.$
    Consequently, $\sum_{j=1}^{r}
    \min\{\nu_{f_0-h}(\alpha_j),s\}
    \le
    \sum_{j=1}^{r}
    \nu_{f_0-h}(\alpha_j)
    \le k$
    and
    \[
    \begin{aligned}
    d_{\operatorname{RT}}\bigl(U_0,\operatorname{Ev}_{\alpha,V}(h)\bigr)
    =
    wt_{\operatorname{RT}}\bigl(\operatorname{Ev}_{\alpha,V}(f_0-h)\bigr)=
    sr-
    \sum_{j=1}^{r}
    \min\{\nu_{f_0-h}(\alpha_j),s\}\ge sr-k.
    \end{aligned}
    \]
 So $$sr-k\leq 
    d_{\operatorname{RT}}(U_0,\operatorname{HRS}_{k}(\alpha,V))\leq \rho_{\operatorname{RT}}(\operatorname{HRS}_{k}(\alpha,V))\leq sr-k.$$ 
    Therefore $\rho_{\operatorname{RT}}(\operatorname{HRS}_{k}(\alpha,V))=sr-k.$
\qed

Remark that for $s=1$ it reduces to the covering radius of GRS codes~\cite[p. 330]{CKM+85}.

For the trivial case $k=sr$, we have $\operatorname{HRS}_{sr}(\alpha,V)=\mathbb{F}_{q}^{s\times r}$. Thus, $\rho_{\operatorname{RT}}(\operatorname{HRS}_{sr}(\alpha,V))=0$. So we always assume $1\leq k\leq sr-1$.

\medskip


For convenience in the sequel, we introduce the following notation. Let \[ \mathcal M_k = \left\{ \boldsymbol{\mu}=(\mu_1,\dots,\mu_r)\in[0,s]^r \ \middle| \ \sum_{j=1}^{r}\mu_j=k+1 \right\}. \] For every \(\boldsymbol{\mu}\in\mathcal M_k\), define $ M_{\boldsymbol{\mu}}(x) = \prod_{j=1}^{r}(x-\alpha_j)^{\mu_j}.$ For a polynomial \(f(x)\), we denote by \([x^k]f(x)\) the coefficient of \(x^k\) in \(f(x)\) and define the linear function \begin{align*}
    L_{\boldsymbol{\mu}}: \mathbb F_q[x]_{\le sr-1} &\longrightarrow \mathbb F_q \\
    f&\longmapsto [x^{k}](f\bmod M_{\boldsymbol{\mu}}).
\end{align*} Let $Q_k= \mathbb F_q[x]_{\le sr-1}/\mathbb F_q[x]_{\le k-1}.$ Since \(L_{\boldsymbol{\mu}}\) vanishes on \(\mathbb F_q[x]_{\le k-1}\), it induces a well-defined linear function \[ \overline{L}_{\boldsymbol{\mu}}: Q_k \longrightarrow \mathbb F_q, \qquad \overline{L}_{\boldsymbol{\mu}} \bigl( f+\mathbb F_q[x]_{\le k-1} \bigr) = L_{\boldsymbol{\mu}}(f). \]

\medskip

By Theorem~\ref{NRT covering radius}, the set of all NRT deep holes of \(\operatorname{HRS}_k(\alpha,V)\) is given as follows:
\[
\operatorname{DH}(\operatorname{HRS}_k(\alpha,V))=\left\{ U\in\mathbb F_q^{s\times r}\mid d_{\operatorname{RT}}\bigl(U,\operatorname{HRS}_k(\alpha,V)\bigr)
= sr-k \right\}.
\]
The following theorem extends the MDS-extension criterion~\cite[Proposition~1]{Kai17} for Reed-Solomon deep holes to HRS codes in the NRT metric.

\begin{Theo}\label{MDS Deep holes}
    Let $U\in\mathbb{F}_{q}^{s\times r}.$ Then $U\in\operatorname{DH}(\operatorname{HRS}_k(\alpha,V))$ if and only if $$\operatorname{HRS}_{k}(\alpha,V)\oplus \mathbb{F}_{q}U$$ is an $[sr,k+1]$ MDS NRT code.
\end{Theo}

{\it Proof.} If $U\in \operatorname{DH}(\operatorname{HRS}_k(\alpha,V))$, then $d_{\operatorname{RT}}(U,\operatorname{HRS}_{k}(\alpha,V))=sr-k\geq1$. This implies $U\notin \operatorname{HRS}_{k}(\alpha,V)$ and $$\dim_{\mathbb{F}_{q}}(\operatorname{HRS}_{k}(\alpha,V)\oplus\mathbb{F}_{q}U)=k+1.$$
Moreover, \begin{align*}
    d_{\operatorname{RT}}(\operatorname{HRS}_{k}(\alpha,V)\oplus \mathbb{F}_{q}U)&= \min\{d_{\operatorname{RT}}(\operatorname{HRS}_{k}(\alpha,V)),d_{\operatorname{RT}}(U,\operatorname{HRS}_{k}(\alpha,V))\}=\min\{sr-k+1,sr-k\}= sr-k.
\end{align*}
Hence $\operatorname{HRS}_{k}(\alpha,V)\oplus \mathbb{F}_{q}U$ is an $[sr,k+1]$ MDS NRT code. Conversely, if $\operatorname{HRS}_{k}(\alpha,V)\oplus \mathbb{F}_{q}U$ is an $[sr,k+1]$ MDS NRT code, then $U\notin \operatorname{HRS}_{k}(\alpha,V)$ and $d_{\operatorname{RT}}(\operatorname{HRS}_{k}(\alpha,V)\oplus \mathbb{F}_{q}U)=sr-k.$ It follows from $d_{\operatorname{RT}}(\operatorname{HRS}_{k}(\alpha,V))=sr-k+1$ that $$d_{\operatorname{RT}}(U,\operatorname{HRS}_{k}(\alpha,V))=sr-k.$$
    This implies $U\in\operatorname{DH}(\operatorname{HRS}_k(\alpha,V)).$\qed

\medskip

The next theorem provides a main explicit criterion for NRT deep holes in this paper, expressing the NRT deep hole condition in terms of polynomial remainders.
\begin{Theo}\label{Deep holes}
    Let $f\in\mathbb{F}_{q}[x].$ Then $\operatorname{Ev}_{\alpha,V}(f)\in\operatorname{DH}(\operatorname{HRS}_k(\alpha,V))$ if and only if \begin{align*}
         \deg(f\bmod M_{\boldsymbol{\mu}})=k \quad \text{for all} \ \boldsymbol{\mu}\in \mathcal{M}_{k}.
     \end{align*}
\end{Theo}

{\it Proof.}
By {\rm Theorem~\ref{NRT covering radius}} and the relationship between NRT deep hole and the order of vanishing, it follows that \(\operatorname{Ev}_{\alpha,V}(f)\in\operatorname{DH}(\operatorname{HRS}_k(\alpha,V))\) if and only if
$$\Delta_k(f):=
\max_{h\in\mathbb F_q[x]_{\le k-1}}
\sum_{j=1}^{r}
\min\{\nu_{f-h}(\alpha_j),s\}=k.$$
Then we may use the same argument as in the proof of Theorem~\ref{NRT covering radius} to show 
$d_{\operatorname{RT}}(\operatorname{Ev}_{\alpha,V}(f),\operatorname{HRS}_{k}(\alpha,V))\leq sr-k$ for all
$f\in\mathbb{F}_{q}[x]$. Indeed, we can choose \(h\in \mathbb F_q[x]_{\le k-1}\) such that $\sum_{j=1}^{r}
\min\{\nu_{f-h}(\alpha_j),s\}\ge k.$
Therefore $\Delta_k(f)\ge k.$
Hence $\Delta_k(f)=k$
is equivalent to saying that there is no \(h\in\mathbb F_q[x]_{\le k-1}\)
such that
\[
\sum_{j=1}^{r}
\min\{\nu_{f-h}(\alpha_j),s\}\ge k+1.
\]
Then we show that this is also equivalent to
\begin{align*}
         \deg(f\bmod M_{\boldsymbol{\mu}})=k \quad \text{for all} \ \boldsymbol{\mu}\in \mathcal{M}_{k}.
     \end{align*}Suppose first that there exists \(h_{0}\in\mathbb F_q[x]_{\le k-1}\) such that $\sum_{j=1}^{r}
\min\{\nu_{f-h_{0}}(\alpha_j),s\}\ge k+1,$
then there exists  $\boldsymbol{\mu}_{0}=(\mu_{0j})_{j\in[r]}\in \mathcal{M}_{k}$ satisfying $0\le \mu_{0j}\le\min\{\nu_{f-h_{0}}(\alpha_j),s\}$ for  $j\in[r].$ This implies $M_{\boldsymbol{\mu}_{0}}(x)
\mid f(x)-h_{0}(x)$ and $$f\bmod M_{\boldsymbol{\mu}_{0}}=h_{0}\in \mathbb F_q[x]_{\le k-1},$$ a contradiction to the assumption that $ \deg(f\bmod M_{\boldsymbol{\mu}})=k $ for $\boldsymbol{\mu}\in \mathcal{M}_{k}.$ Conversely, suppose that there exists  $\boldsymbol{\mu}_{1}=(\mu_{1j})_{j\in[r]}\in\mathcal M_k$ satisfying $f\bmod M_{\boldsymbol{\mu}_{1}}\in \mathbb F_q[x]_{\le k-1}$. Let $h_{1}=f\bmod M_{\boldsymbol{\mu}_{1}}\in\mathbb F_q[x]_{\leq k-1}.$
Then $(x-\alpha_j)^{\mu_{1j}}\mid f(x)-h_{1}(x)$ for $j\in[r]$ which implies $\nu_{f-h_{1}}(\alpha_j)\ge \mu_{1j}.$
Since \(\mu_{1j}\le s\), it follows that $\min\{\nu_{f-h_{1}}(\alpha_j),s\}\ge \mu_{1j}$ and
\[
\sum_{j=1}^{r}
\min\{\nu_{f-h_{1}}(\alpha_j),s\}
\ge
\sum_{j=1}^{r}\mu_{1j}
=
k+1.
\]
But this contradicts the assumption that there is no \(h\in\mathbb F_q[x]_{\le k-1}\)
such that $\sum_{j=1}^{r}
\min\{\nu_{f-h}(\alpha_j),s\}\ge k+1.$ This finishes the proof.\qed

\medskip

The following corollary reformulates the polynomial criterion as a complement-counting problem for a hyperplane arrangement in $Q_{k}=\mathbb F_q[x]_{\le sr-1}/\mathbb F_q[x]_{\le k-1}$. Now, we review hyperplane arrangement and its characteristic polynomial in~\cite{OT92,Sta07}.

Let $\mathbb{K}$ be a field and let $V$ be a $\mathbb{K}$-vector space of dimension $\ell$.  A hyperplane $H$ in $V$ is an affine subspace of dimension $\ell-1$. A hyperplane arrangement $\mathcal{A}$ is a finite set of hyperplanes in $V$. The characteristic polynomial of $\mathcal{A}$, denoted by $\chi_{\mathcal{A}}$, is defined to be $$\chi_{\mathcal{A}}(t):=\sum_{X\in L(\mathcal{A})}\mu(X)t^{\dim(X)},$$
where $L(\mathcal{A})$ is the set of nonempty intersections of elements of $\mathcal{A}$, including $V$  itself as the intersection over the empty set and $\mu$ is the M$\ddot{\text{o}}$bius function. Any subset $\mathcal{B}$ of $\mathcal{A}$ is called central if $\cap_{H\in\mathcal{B}}H\neq \emptyset$. H. Whitney proved in~\cite[p. 573]{Whi32} that$$\chi_{\mathcal{A}}(t)=\sum_{\substack{\mathcal{B}\subseteq\mathcal{A}\\\mathcal{B}\, \text{central}}}(-1)^{\left|\mathcal{B}\right|}t^{\dim(\cap_{H\in\mathcal{B}}H)}.$$

\begin{Coro}\label{Count Deep holes}
   Let $\mathcal{A}_{s,r,k}=\{\Ker(\overline{L}_{\boldsymbol{\mu}})\mid \boldsymbol{\mu}\in\mathcal{M}_{k}\}.$ Then $$\operatorname{DH}(\operatorname{HRS}_{k}(\alpha,V))=\{U\in\mathbb{F}_{q}^{s\times r}\mid L_{\boldsymbol{\mu}}(f_{U})\neq 0 \ \text{for all} \ \boldsymbol{\mu}\in \mathcal{M}_{k}\}.$$
    Consequently, $$\# \operatorname{DH}(\operatorname{HRS}_{k}(\alpha,V))=q^{k}\cdot\chi_{\mathcal{A}_{s,r,k}}(q).$$
\end{Coro}

{\it Proof.} Since $\deg M_{\boldsymbol{\mu}}=k+1,$ we get $\deg(f_{U} \bmod M_{\boldsymbol{\mu}})=k$ is equivalent to $L_{\boldsymbol{\mu}}(f_{U})\neq 0.$ Thus it follows from {\rm Theorem~\ref{Deep holes}} that $$\operatorname{DH}(\operatorname{HRS}_{k}(\alpha,V))=\{U\in\mathbb{F}_{q}^{s\times r}\mid L_{\boldsymbol{\mu}}(f_{U})\neq 0\ \text{for all} \ \boldsymbol{\mu}\in \mathcal{M}_{k}\}.$$
Since \(\mathbb F_q[x]_{\le k-1}\subseteq \Ker(L_{\boldsymbol{\mu}})\), we have that $U\in\operatorname{DH}(\operatorname{HRS}_{k}(\alpha,V))$
if and only if $f_U+\mathbb F_q[x]_{\le k-1}
\in
Q_k
\setminus
\bigcup_{\boldsymbol{\mu}\in\mathcal M_k}
\Ker(\overline L_{\boldsymbol{\mu}}).$
This implies that $\# \operatorname{DH}(\operatorname{HRS}_{k}(\alpha,V))
=
q^k\cdot
\#\left(
Q_k
\setminus
\bigcup_{\boldsymbol{\mu}\in\mathcal M_k}
\Ker(\overline L_{\boldsymbol{\mu}})
\right).$
By the inclusion-exclusion principle, we have \begin{align*}
    \#\left( Q_k \setminus \bigcup_{\boldsymbol{\mu}\in\mathcal M_k} \Ker(\overline L_{\boldsymbol{\mu}}) \right) &= \#\left( Q_k \setminus \bigcup_{H\in\mathcal A_{s,r,k}}H \right)= \sum_{\mathcal B\subseteq \mathcal A_{s,r,k}} (-1)^{|\mathcal B|} \cdot\#\left(\bigcap_{H\in\mathcal B}H\right)\\
   & = \sum_{\mathcal B\subseteq \mathcal A_{s,r,k}} (-1)^{|\mathcal B|} \cdot q^{\dim_{\mathbb F_q}\left(\bigcap_{H\in\mathcal B}H\right)}= \chi_{\mathcal A_{s,r,k}}(q).
\end{align*}
Thus, we have $ \# \operatorname{DH}(\operatorname{HRS}_{k}(\alpha,V)) = q^k\cdot\chi_{\mathcal A_{s,r,k}}(q).$\qed

The above result can be reformulated as the following geometric problem:
\begin{align*}
   \text{how many hyperplanes in} \ \operatorname{PG}(Q_{k}^{*}) \ \text{ are disjoint from} \ \{[\overline{L}_{\boldsymbol{\mu}}]\mid \boldsymbol{\mu}\in\mathcal{M}_{k}\},
\end{align*}
where $Q_{k}^{*}=\Hom_{\mathbb{F}_{q}}(Q_{k},\mathbb{F}_{q})$ is the dual of the vector space $Q_{k}$. 

\begin{Prop}
    If $1\le k\leq sr-2$, then\begin{align*}
    \chi_{\mathcal{A}_{s,r,k}}(q)=(q-1)\cdot\#\left\{\Pi\subseteq\operatorname{PG}(Q_{k}^{*})\middle| \begin{array}{l}
\dim(\Pi)=sr-k-2\\[1mm]
\Pi\cap\{[\overline{L}_{\boldsymbol{\mu}}]\mid \boldsymbol{\mu}\in\mathcal{M}_{k}\}=\emptyset
\end{array}\right\}.
\end{align*}
\end{Prop}

\textit{Proof.}
Define a map \begin{align*}
    \pi:Q_k \setminus \bigcup_{\boldsymbol{\mu}\in\mathcal M_k} \Ker(\overline L_{\boldsymbol{\mu}})&\longrightarrow\left\{\Pi\subseteq\operatorname{PG}(Q_{k}^{*})\middle| \begin{array}{l}
\dim(\Pi)=sr-k-2\\[1mm]
\Pi\cap\{[\overline{L}_{\boldsymbol{\mu}}]\mid \boldsymbol{\mu}\in\mathcal{M}_{k}\}=\emptyset
\end{array}\right\}.\\
   f+\mathbb{F}_{q}[x]_{\leq k-1} &\longmapsto \Pi_{f}:=\{[\lambda]\in \operatorname{PG}(Q_{k}^{*})\mid \lambda(f+\mathbb{F}_{q}[x]_{\leq k-1})=0\}.
\end{align*}
\begin{enumerate}
    \item well-defined: Note that $$\dim_{\mathbb{F}_{q}}\{\lambda\in Q_{k}^{*}\mid \lambda(f+\mathbb{F}_{q}[x]_{\leq k-1})=0\}=\dim_{\mathbb{F}_{q}}(Q_{k}^{*})-1=sr-k-1$$ for all $f+\mathbb{F}_{q}[x]_{\leq k-1}\in Q_k \setminus \bigcup_{\boldsymbol{\mu}\in\mathcal M_k} \Ker(\overline L_{\boldsymbol{\mu}})$. This implies $\dim(\Pi_{f})=sr-k-2$. Suppose that there exists $\boldsymbol{\mu}_{0}\in\mathcal{M}_{k}$ such that $[\overline{L}_{\boldsymbol{\mu}_{0}}]\in \Pi_{f}$, then $\overline{L}_{\boldsymbol{\mu}_{0}}(f+\mathbb{F}_{q}[x]_{\leq k-1})=0$. This is a contradiction. So $\pi$ is well-defined.
    \item surjective: Given an element $\Pi\in\{\Pi\subseteq\operatorname{PG}(Q_{k}^{*})\mid \dim(\Pi)=sr-k-2,\Pi\cap\{[\overline{L}_{\boldsymbol{\mu}}]\mid \boldsymbol{\mu}\in\mathcal{M}_{k}\}=\emptyset\}$, there exists a hyperplane $H_{0}$ of $Q_{k}^{*}$ such that $\Pi=\operatorname{PG}(H_{0})$. This implies there exists $0\neq F_{0}\in Q_{k}^{**}$ such that $H_{0}=\Ker(F_{0})$. Since the map $\theta:Q_{k}\rightarrow Q_{k}^{**}$ given by $f+\mathbb{F}_{q}[x]_{\leq k-1}\mapsto[\lambda\mapsto\lambda(f+\mathbb{F}_{q}[x]_{\leq k-1})]$ is an $\mathbb{F}_{q}$-linear isomorphism, there exists a unique element $f_{0}+\mathbb{F}_{q}[x]_{\leq k-1}\in Q_{k}$ such that $F_{0}=\theta(f_{0}+\mathbb{F}_{q}[x]_{\leq k-1})$. Thus we have $\Pi=\{[\lambda]\in\operatorname{PG}(Q_{k}^{*})\mid \lambda(f_{0}+\mathbb{F}_{q}[x]_{\leq k-1})=0\}.$ It follows from $\Pi\cap\{[\overline{L}_{\boldsymbol{\mu}}]\mid \boldsymbol{\mu}\in\mathcal{M}_{k}\}=\emptyset$ that $f_{0}+\mathbb{F}_{q}[x]_{\leq k-1}\notin \bigcup_{\boldsymbol{\mu}\in\mathcal M_k} \Ker(\overline L_{\boldsymbol{\mu}})$.
    \item $\pi^{-1}(\Pi_{f})=\{c(f+\mathbb{F}_{q}[x]_{\leq k-1})\mid c\in\mathbb{F}_{q}^{*}\}$: If $\Pi_{f_{1}}=\Pi_{f_{2}}$, then we have $\Ker(\theta(f_{1}+\mathbb{F}_{q}[x]_{\leq k-1}))=\Ker(\theta(f_{2}+\mathbb{F}_{q}[x]_{\leq k-1})).$ Let $F_{i}=\theta(f_{i}+\mathbb{F}_{q}[x]_{\leq k-1})$ for $i=1,2$. Since $F_{1},F_{2}\neq 0$ and $\Ker(F_{1})=\Ker(F_{2})$, there exists $\lambda_{0}\in Q_{k}^{*}$ such that $F_{1}(\lambda_{0})\neq 0$ and $F_{2}(\lambda_{0})\neq 0.$ Note that $\lambda-\frac{F_{2}(\lambda)}{F_{2}(\lambda_{0})}\lambda_{0}\in\Ker(F_{2})=\Ker(F_{1})$ for all $\lambda\in Q_{k}^{*}$. This implies $F_{1}(\lambda)=\frac{F_{1}(\lambda_{0})}{F_{2}(\lambda_{0})}F_{2}(\lambda)$ and $f_{1}+\mathbb{F}_{q}[x]_{\leq k-1}=\frac{F_{1}(\lambda_{0})}{F_{2}(\lambda_{0})}(f_{2}+\mathbb{F}_{q}[x]_{\leq k-1})$ since $\theta$ is an $\mathbb{F}_{q}$-linear isomorphism. Hence $\pi^{-1}(\Pi_{f})\subseteq\{c(f+\mathbb{F}_{q}[x]_{\leq k-1})\mid c\in\mathbb{F}_{q}^{*}\}$. Conversely, since $\lambda\in Q_{k}^{*}$ is a $\mathbb{F}_{q}$-linear map, $\lambda(f+\mathbb{F}_{q}[x]_{\leq k-1})=0$ if and only if $\lambda(cf+\mathbb{F}_{q}[x]_{\leq k-1})=0$ for any $c\in\mathbb{F}_{q}^{*}$. Thus $\pi^{-1}(\Pi_{f})\supseteq\{c(f+\mathbb{F}_{q}[x]_{\leq k-1})\mid c\in\mathbb{F}_{q}^{*}\}$.
\end{enumerate}
This finishes the proof of the above equality for $\chi_{\mathcal{A}_{s,r,k}}(q).$\qed

The preceding reformulation connects the deep-hole
problem with counting questions studied in finite
geometry.
Metsch~\cite{Met07} investigated lower bounds for
the number of subspaces disjoint from a point set
of prescribed size.
The number of disjoint hyperplanes is also the
zero-intersection entry of the hyperplane spectrum
of a point set.
Standard counting identities for this spectrum
and its relation to the weight distribution of
projective codes are discussed in~\cite{HKK18}.
In our setting, the relevant point set is
\[
\{[\overline L_{\boldsymbol\mu}]
  \mid \boldsymbol\mu\in\mathcal M_k\}
\subseteq \operatorname{PG}(Q_k^*),
\]
whose structure is determined by the evaluation
points and the admissible multiplicities.
This provides a finite-geometric perspective on
the enumeration of NRT deep holes.

\section{Some explicit deep holes of HRS codes}\label{Some explicit constructions of NRT deep holes}

In this section, 
We give three explicit families of NRT deep holes. The first consists of evaluations of  degree-$k$ polynomials; the second generalizes reciprocal-type deep holes of Reed-Solomon codes; the third gives a family in the low-layer range $1\leq k\leq s-1$. We later prove that the third family is the complete set of deep holes for full-length evaluation sets in the range stated in Proposition~\ref{Conjecture2}.
\begin{Theo}\label{When Ev(f) deep hole}
\begin{enumerate}[label=\textup{(\arabic*)}]
    \item  If $\deg (f(x))=k$, then $\operatorname{Ev}_{\alpha,V}(f)\in\operatorname{DH}(\operatorname{HRS}_k(\alpha,V)).$   Moreover, if $f(x)=ax^{k}+g(x)$ where $a\in \mathbb{F}_{q}^{*}$, $g\in\mathbb{F}_{q}[x]_{\leq k-1}$, then the set of all nearest codewords to $\operatorname{Ev}_{\alpha,V}(f)$ is $$\{\operatorname{Ev}_{\alpha,V}(h_{\boldsymbol{\nu}})\mid h_{\boldsymbol{\nu}}(x)=g(x)+a(x^{k}-M_{\boldsymbol{\nu}}(x)),\boldsymbol{\nu}\in\mathcal{M}_{k-1}\}.$$
    \item If there exist $\beta\in\mathbb{F}_{q}\backslash\{\alpha_{1},\dots,\alpha_{r}\}$, $a\in\mathbb{F}_{q}^{*}$ and $g\in\mathbb{F}_{q}[x]_{\leq k-1}$ such that $$(x-\beta)(f(x)-g(x))\equiv a \bmod \pi_{\alpha}(x)^{s},$$ then $$\operatorname{Ev}_{\alpha,V}(f)\in\operatorname{DH}(\operatorname{HRS}_k(\alpha,V)).$$
    \item Suppose that $s\geq 2$ and $1\leq k \leq s-1$. 
Let $$\mathcal{D}_{k,s}=\left\{\operatorname{Ev}_{\alpha,V}(ax^{k}+g(x)+\pi_{\alpha}(x)^{k}h(x))\ \middle| \ \begin{array}{c}
       a\in\mathbb{F}_{q}^{*},  g(x)\in\mathbb{F}_{q}[x]_{\leq k-1}\\
h(x)\in\mathbb{F}_{q}[x]_{\leq(s-k)r-1}   \\
      \text{with}\  h(\alpha_{j})\neq -a\eta_{j}^{-k}\ \text{for} \ j\in[r]
    \end{array}\right\}.$$ Then $$\mathcal{D}_{k,s}\subseteq\operatorname{DH}(\operatorname{HRS}_{k}(\alpha,V))$$ and$$\#\mathcal{D}_{k,s}=\begin{cases}
    (q-1)q^{s-1}, & r=1;\\
        (q-1)^{r+1}q^{k+(s-k-1)r}, & r\geq 2.
    \end{cases}$$
\end{enumerate}
\end{Theo}

{\it Proof.} \begin{enumerate}[label=\textup{(\arabic*)}]
    \item The deep-hole assertion follows from Theorem\ref{Deep holes}. Write $f(x)=ax^{k}+g(x)$ where $a\in \mathbb{F}_{q}^{*}$, $g\in\mathbb{F}_{q}[x]_{\leq k-1}$. For  $\boldsymbol{\nu}=(\nu_{j})_{j\in[r]}\in\mathcal{M}_{k-1}$, $h_{\boldsymbol{\nu}}(x)
:=
g(x)+a(x^k-M_{\boldsymbol{\nu}}(x))
\in\mathbb F_q[x]_{\le k-1}$.
So \(\operatorname{Ev}_{\alpha,V}(h_{\boldsymbol{\nu}})\in\operatorname{HRS}_k(\alpha,V)\).
Since
\begin{align*}
    d_{\operatorname{RT}}\bigl(
\operatorname{Ev}_{\alpha,V}(f),
\operatorname{Ev}_{\alpha,V}(h_{\boldsymbol{\nu}})
\bigr)
=wt_{\operatorname{RT}}\bigl(\operatorname{Ev}_{\alpha,V}(f-h_{\boldsymbol{\nu}})\bigr) =wt_{\operatorname{RT}}\bigl(\operatorname{Ev}_{\alpha,V}(aM_{\boldsymbol{\nu}})\bigr)=sr-\sum_{j=1}^{r}\nu_j =sr-k,
\end{align*}
the codeword \(\operatorname{Ev}_{\alpha,V}(h_{\boldsymbol{\nu}})\) is a nearest
codeword to \(\operatorname{Ev}_{\alpha,V}(f)\). Conversely, if \(\operatorname{Ev}_{\alpha,V}(h)\) is a nearest codeword to
\(\operatorname{Ev}_{\alpha,V}(f)\), where $h\in\mathbb F_q[x]_{\le k-1}$,
then
$\sum_{j=1}^{r}\min\{\nu_{f-h}(\alpha_{j}),s\}=k.$
Since \(\deg(f-h)=k\),
\begin{align*}
    k=\sum_{j=1}^{r}\min\{\nu_{f-h}(\alpha_{j}),s\}\leq \sum_{j=1}^{r}\nu_{f-h}(\alpha_{j})\leq \deg(f-h)=k,
\end{align*}
which implies $\boldsymbol{\nu}_{f-h}=(\nu_{f-h}(\alpha_{1}),\dots,\nu_{f-h}(\alpha_{r}))\in \mathcal{M}_{k-1}$ and $a^{-1}(f(x)-h(x))=M_{\boldsymbol{\nu}_{f-h}}(x).$
Thus
\[
h(x)=f(x)-aM_{\boldsymbol{\nu}_{f-h}}(x)=g(x)+a(x^k-M_{\boldsymbol{\nu}_{f-h}}(x))=h_{\boldsymbol{\nu}_{f-h}}.
\]
Thus the set of all nearest codewords to
\(\operatorname{Ev}_{\alpha,V}(f)\) is precisely
\[
\left\{
\operatorname{Ev}_{\alpha,V}(h_{\boldsymbol{\nu}})
\mid
h_{\boldsymbol{\nu}}(x)
=
g(x)+a(x^k-M_{\boldsymbol{\nu}}(x)),
\ 
\boldsymbol{\nu}\in\mathcal M_{k-1}
\right\}.
\]
    \item   
    Suppose that $\operatorname{Ev}_{\alpha,V}(f)\notin \operatorname{DH}(\operatorname{HRS}_k(\alpha,V))$, by {\rm Theorem~\ref{Deep holes}}, there exists $\boldsymbol{\mu}_{0}=(\mu_{0j})_{j\in[r]}\in \mathcal{M}_{k}$ such that $\deg(f \bmod M_{\boldsymbol{\mu}_{0}})\leq k-1.$
   Let $h_{0}=f\bmod M_{\boldsymbol{\mu}_{0}}\in\mathbb{F}_{q}[x]_{\leq k-1}.$ 
Since $M_{\boldsymbol{\mu}_{0}}(x)\mid \pi_{\alpha}(x)^{s},$ we have $$(x-\beta)(h_{0}(x)-g(x))\equiv (x-\beta)(f(x)-g(x))\equiv a\bmod M_{\boldsymbol{\mu}_{0}}(x).$$  Note that $\deg((x-\beta)(h_{0}-g)-a)\leq k<k+1=\deg(M_{\boldsymbol{\mu}_{0}})$. This implies $$(x-\beta)(h_{0}(x)-g(x))=a$$ which contradicts the assumption for $a\in\mathbb{F}_{q}^{*}$.  Thus $\operatorname{Ev}_{\alpha,V}(f)\in\operatorname{DH}(\operatorname{HRS}_k(\alpha,V)).$
\item Let $F(x)=ax^{k}+g(x)+\pi_{\alpha}(x)^{k}h(x)\in\mathbb{F}_{q}[x]_{\leq sr-1}.$ By Theorem~\ref{Deep holes}, it suffices to demonstrate that $\deg(F\bmod M_{\boldsymbol{\mu}})=k$ for $ \boldsymbol{\mu}\in \mathcal{M}_{k}.$ If $\mu_{j}\leq k $ for $j\in[r]$, then $M_{\boldsymbol{\mu}}(x)\mid \pi_{\alpha}(x)^{k}$ and $F(x)\equiv ax^{k}+g(x)\bmod M_{\boldsymbol{\mu}}(x).$ Thus
$\deg(F\bmod M_{\boldsymbol{\mu}})=k.$ Otherwise, we suppose that there exists a unique $j_{0}\in[r]$ such that $\mu_{j_{0}}=k+1$ and $\mu_{j}=0$ for $j\neq j_{0}.$ Thus $M_{\boldsymbol{\mu}}=(x-\alpha_{j_{0}})^{k+1}$. We claim that $L_{\boldsymbol{\mu}}(F)=a+\eta_{j_{0}}^{k}h(\alpha_{j_{0}}).$ Clearly, we have $\prod_{\ell\neq j_{0}}(x-\alpha_{\ell})\equiv\eta_{j_{0}}\bmod (x-\alpha_{j_{0}}).$ This implies that $\prod_{\ell\neq j_{0}}(x-\alpha_{\ell})=f_{0}(x)(x-\alpha_{j_{0}})+\eta_{j_{0}}$ for some $f_{0}(x)\in\mathbb{F}_{q}[x]$. Then we have \begin{align*}
    \pi_{\alpha}(x)^{k} &= (f_{0}(x)(x-\alpha_{j_{0}})+\eta_{j_{0}})^{k}(x-\alpha_{j_{0}})^{k}\\
    &=g_{0}(x)(x-\alpha_{j_{0}})^{k+1}+\eta_{j_{0}}^{k}(x-\alpha_{j_{0}})^{k}\\
    &\equiv \eta_{j_{0}}^{k}(x-\alpha_{j_{0}})^{k}\bmod M_{\boldsymbol{\mu}}
\end{align*}
for some $g_{0}(x)\in\mathbb{F}_{q}[x]$. We may write $h(x)=h(\alpha_{j_{0}})+(x-\alpha_{j_{0}})h_{0}(x)$ for some $h_{0}(x)\in\mathbb{F}_{q}[x]$. Thus $\pi_{\alpha}(x)^{k}h(x)\equiv\eta_{j_{0}}^{k}h(\alpha_{j_{0}})(x-\alpha_{j_{0}})^{k} \bmod M_{\boldsymbol{\mu}}$ and $L_{\boldsymbol{\mu}}(F)=a+\eta^{k}_{j_{0}}h(\alpha_{j_{0}})$. By the definition of $\mathcal{D}_{k,s}$, we have $a+\eta_{j_{0}}^{k}h(\alpha_{j_{0}})\neq 0$  and $\deg(F \bmod M_{\boldsymbol{\mu}})=k$. So $\mathcal{D}_{k,s}\subseteq \operatorname{DH}(\operatorname{HRS}_{k}(\alpha,V)).$

If $r=1$, we have $\pi_{\alpha}(x)=x-\alpha_{1}$, $\eta_{1}=1$.
We claim that \begin{gather*}
    \left\{
        ax^{k}+g(x)+\pi_{\alpha}(x)^{k}h(x)
        \;\middle|\;
        \begin{array}{c}
            a\in\mathbb{F}_{q}^{*},\quad
            g(x)\in\mathbb{F}_{q}[x]_{\leq k-1},\\
            h(x)\in\mathbb{F}_{q}[x]_{\leq s-k-1}\text{ with }h(\alpha_{1})\neq -a
        \end{array}
    \right\}
    =
    \left\{
        \sum_{i=0}^{s-1}c_{i}(x-\alpha_{1})^{i}
        \;\middle|\;
        c_{k}\neq 0
    \right\}.
\end{gather*}
 By the preceding argument, it follows  that $\deg(F\bmod (x-\alpha_{1})^{k+1})=k$ and $F(x)\in\left\{\sum_{i=0}^{s-1}c_{i}(x-\alpha_{1})^{i}\ \middle| \ c_{k}\neq 0\right\}.$ Conversely, if $c(x)=\sum_{i=0}^{s-1}c_{i}(x-\alpha_{1})^{i}$ with $c_{k}\neq0$, then there exist $g(x)\in\mathbb{F}_{q}[x]_{\leq k-1}$ and $h(x)\in\mathbb{F}_{q}[x]_{\leq s-k-1}$ such that $c(x)=x^{k}+g(x)+(x-\alpha_{1})^{k}h(x)$. Consider the coefficient of $(x-\alpha_{1})^{k}$ in $c(x)$. Then $1+h(\alpha_{1})=c_{k}\neq 0.$ Thus $$c(x)\in\left\{ax^{k}+g(x)+(x-\alpha_{1})^{k}h(x)\ \middle| \ \begin{array}{c}
       a\in\mathbb{F}_{q}^{*}, g(x)\in\mathbb{F}_{q}[x]_{\leq k-1} \\ h(x)\in\mathbb{F}_{q}[x]_{\leq s-k-1}  \text{ with}\   
        h(\alpha_{1})\neq -a
    \end{array}\right\}.$$ It is easy to see that $\#\mathcal{D}_{k,s}=(q-1)q^{s-1}$ for $r=1.$

    Now, we claim that if $\operatorname{Ev}_{\alpha,V}(ax^{k}+g(x)+\pi_{\alpha}(x)^{k}h(x))=\operatorname{Ev}_{\alpha,V}(a'x^{k}+g'(x)+\pi_{\alpha}(x)^{k}h'(x)),$ then $(a,g,h)=(a',g',h').$ By Lemma~\ref{iso}, we have $(a-a')x^{k}+g(x)-g'(x)=\pi_{\alpha}(x)^{k}(h'(x)-h(x)).$
    Note that $\deg((a-a')x^{k}+g(x)-g'(x))\leq k$. This implies $h'=h$. Otherwise, we would have $k\geq\deg((a-a')x^{k}+g(x)-g'(x))=\deg(\pi_{\alpha}(x)^{k}(h'(x)-h(x)))\geq 2k$ since $r\geq 2$.
    Since $\deg(g-g')\leq k-1$, we have $a=a'$ and $g=g'$.  Then we show that the number of admissible triples $(a,g,h)$ is $(q-1)^{r+1}q^{k+(s-k-1)r}$. By construction, there are $q-1$ choices for $a\in\mathbb{F}_{q}^{*}$ and $q^{k}$ choices for $g(x)\in\mathbb{F}_{q}[x]_{\leq k-1}.$ Since $s-k\geq 1$, the linear evaluation map $\mathbb{F}_{q}[x]_{\leq(s-k)r-1}\rightarrow \mathbb{F}_{q}^{r}$, defined by $h(x)\mapsto (h(\alpha_{j}))_{j\in[r]}$, is surjective and every vector in $\mathbb{F}_{q}^{r}$ has exactly $q^{(s-k)r}/q^{r}=q^{(s-k-1)r}$ preimages. For any $a\in\mathbb{F}_{q}^{*}$, there are $(q-1)^{r}$ possible values of $(h(\alpha_{j}))_{j\in[r]}$. This implies also that there are $(q-1)^{r}q^{(s-k-1)r}$ choices for $h(x)\in\mathbb{F}_{q}[x]_{\leq (s-k)r-1}$ for any $a\in\mathbb{F}_{q}^{*}$. Thus $$\#\mathcal{D}_{k,s}=(q-1)\cdot q^{k}\cdot (q-1)^{r}q^{(s-k-1)r}=(q-1)^{r+1}q^{k+(s-k-1)r}\quad \text{for}\ r\geq 2.$$\qed
\end{enumerate}

The next result follows immediately from  Corollary {\rm \ref{Count Deep holes}} and Theorem {\rm \ref{When Ev(f) deep hole}\textup{(1)}}.
\begin{Coro}
     $\# \operatorname{DH}(\operatorname{HRS}_{k}(\alpha,V))\ge \max\{q^{sr}-B_{s,r}(k+1)q^{sr-1},(q-1)q^{k}\}$
    where $B_{s,r}(k+1)=[z^{k+1}](1+z+\cdots+z^{s})^{r}=\sum_{i=0}^{\lfloor \frac{k+1}{s+1}\rfloor}(-1)^{i}\binom{r}{i}\binom{k+r-i(s+1)}{r-1}.$
\end{Coro}

\section{Enumeration of NRT deep holes in special cases}\label{Enumeration of NRT deep holes in special cases}
In this section, we derive explicit enumeration formulas for the NRT deep holes of $\operatorname{HRS}_{k}(\alpha,V)$ in several special cases, namely $r=1,2$, $k=1,2$ or $k=sr-1,sr-2,sr-3.$ 

\medskip

\begin{Prop}
    If $r=1$, then $$\operatorname{DH}(\operatorname{HRS}_k(\alpha,V))=\{U\in\mathbb{F}_{q}^{s\times 1}\mid U_{k+1,1}\neq 0\},$$
    and $\#\operatorname{DH}(\operatorname{HRS}_k(\alpha,V))=(q-1)q^{s-1}.$
\end{Prop}

{\it Proof.} Note that under the assumption, $\mathcal{M}_{k}=\{(k+1)\}.$ By {\rm Theorem~\ref{Deep holes}}, $U\in\operatorname{DH}(\operatorname{HRS}_k(\alpha,V))$ if and only if $\deg(f_{U}\bmod (x-\alpha_{1})^{k+1})=k.$ By the definition $f_{U}\bmod M_{\boldsymbol{\mu}}$, the Hasse-Taylor expansion of $f_{U}\bmod M_{\boldsymbol{\mu}}$ at $\alpha_{1}$ gives $f_{U} \bmod M_{\boldsymbol{\mu}}=\sum_{i=0}^{k}\partial^{(i)}f_{U}(\alpha_{1})(x-\alpha_{1})^{i}$. Therefore, the above equivalent condition holds if and only if  $\partial^{(k)}f_{U}(\alpha_{1})\neq 0,$ or equivalently, $$U_{k+1,1}=v_{k+1,1}\partial^{(k)}f_{U}(\alpha_{1})\neq 0.$$
    Thus $\operatorname{DH}(\operatorname{HRS}_k(\alpha,V))=\{U\in\mathbb{F}_{q}^{s\times 1}\mid U_{k+1,1}\neq 0\}$
    and $\#\operatorname{DH}(\operatorname{HRS}_k(\alpha,V))=(q-1)q^{s-1}.$\qed

\medskip

Before stating our result for $r=2$, we first present two lemmas that will be used in the proof. 
\begin{Lem}\label{Generalization Vandermonde determinant}\cite[Proposition 1]{Kra01}
    Let $X_{1},X_{2},\cdots,X_{n}$ be indeterminates. If $p_{1},p_{2},\dots,p_{n}$ are polynomials of the form $p_{j}(x)=a_{j}x^{j-1}+$lower terms, then $$\det_{1\leq i,j\leq n}(p_{j}(X_{i}))=a_{1}a_{2}\cdots a_{n}\prod_{1\leq i<j\leq n}(X_{j}-X_{i}).$$
\end{Lem}

For $r=2$, we define\begin{align*}
    h_{s,k}:=\left|\mathcal{M}_{k}\right|=\begin{cases}
            k+2, &1\leq k\leq s-1;\\
            2s-k, &s\leq k \leq 2s-1.
        \end{cases}
    \end{align*} and list the elements of $\mathcal{M}_{k}$ as $\mathcal{M}_{k}=\{\boldsymbol{\mu}^{(0)},\dots,\boldsymbol{\mu}^{(h_{s,k}-1)}\}$ in increasing order of their first coordinates, where $\boldsymbol{\mu}^{(i)}=(\mu_{1}^{(i)},k+1-\mu_{1}^{(i)})$ for $i\in[0,h_{s,k}-1].$

\begin{Lem}\label{Lmu nonsingular}
    Let $e_{\ell}=(x-\alpha_{1})^{k+\ell}+\mathbb{F}_{q}[x]_{\leq k-1}\in Q_{k}$ for $\ell \in[0,h_{s,k}-1]$. If $r=2$, then $$\overline{L}_{\boldsymbol{\mu}^{(i)}}(e_{\ell})=(\alpha_{2}-\alpha_{1})^{\ell}\binom{k+\ell-\mu_{1}^{(i)}}{\ell}\quad \text{for } \ell,i\in[0,h_{s,k}-1].$$Moreover, the matrix $\left(\overline{L}_{\boldsymbol{\mu}^{(i)}}(e_{\ell})\right)_{0\leq i,\ell\leq h_{s,k}-1}$ is nonsingular.
\end{Lem}

{\it Proof.} By definition, we have
\begin{align*}
    \overline{L}_{\boldsymbol{\mu}^{(i)}}(e_{\ell})=[x^{k}]((x-\alpha_{1})^{k+\ell}\bmod (x-\alpha_{1})^{\mu_{1}^{(i)}}(x-\alpha_{2})^{k+1-\mu_{1}^{(i)}}).
\end{align*}

Now, let $R(x)=(x-\alpha_{1})^{k+\ell}\bmod (x-\alpha_{1})^{\mu_{1}^{(i)}}(x-\alpha_{2})^{k+1-\mu_{1}^{(i)}} \in\mathbb{F}_{q}[x]_{\leq k}.$
\begin{enumerate}
    \item If $\mu_{1}^{(i)}=k+1$, then $$R(x)=(x-\alpha_{1})^{k+\ell}\bmod (x-\alpha_{1})^{k+1}=\begin{cases}
       (x-\alpha_{1})^{k}, & \ell=0; \\
        0,    & \ell\geq 1.
    \end{cases}$$ 
    This implies$$\overline{L}_{\boldsymbol{\mu}^{(i)}}(e_{\ell})=\begin{cases}
        1, &\ell=0;\\
        0, &\ell\geq1.
    \end{cases}$$ This agrees with the desired formula. Indeed, since
$\mu_1^{(i)}=k+1$, we have $(\alpha_2-\alpha_1)^\ell
\binom{k+\ell-\mu_1^{(i)}}{\ell}
=
(\alpha_2-\alpha_1)^\ell
\binom{\ell-1}{\ell},$
which equals $1$ for $\ell=0$, by the convention
$\binom{-1}{0}=1$, and equals $0$ for $\ell\ge1$.
    \item If $\mu_{1}^{(i)}\leq k$, then it follows from $(x-\alpha_{1})^{\mu_{1}^{(i)}}\mid(x-\alpha_{1})^{k+\ell}$ that $(x-\alpha_{1})^{\mu_{1}^{(i)}}\mid R(x)$. This yields $$R(x)=(x-\alpha_{1})^{\mu_{1}^{(i)}}A(x)$$ for some $A(x)\in\mathbb{F}_{q}[x]_{\leq k-\mu_{1}^{(i)}}.$ Then \begin{align*}
        A(x)&= (x-\alpha_{1})^{k+\ell-\mu_{1}^{(i)}} \bmod (x-\alpha_{2})^{k+1-\mu_{1}^{(i)}}\\
        &= (x-\alpha_{2}+\alpha_{2}-\alpha_{1})^{k+\ell-\mu_{1}^{(i)}} \bmod (x-\alpha_{2})^{k+1-\mu_{1}^{(i)}}\\
        &=\sum_{j=0}^{k+\ell-\mu_{1}^{(i)}}\binom{k+\ell-\mu_{1}^{(i)}}{j}(\alpha_{2}-\alpha_{1})^{j}(x-\alpha_{2})^{k+\ell-\mu_{1}^{(i)}-j} \bmod (x-\alpha_{2})^{k+1-\mu_{1}^{(i)}}\\
        &=\sum_{j=\ell}^{k+\ell-\mu_{1}^{(i)}}\binom{k+\ell-\mu_{1}^{(i)}}{j}(\alpha_{2}-\alpha_{1})^{j}(x-\alpha_{2})^{k+\ell-\mu_{1}^{(i)}-j}.
    \end{align*}
   Thus $R(x)=(x-\alpha_{1})^{\mu_{1}^{(i)}}A(x)=\sum_{j=\ell}^{k+\ell-\mu_{1}^{(i)}}\binom{k+\ell-\mu_{1}^{(i)}}{j}(\alpha_{2}-\alpha_{1})^{j}(x-\alpha_{1})^{\mu_{1}^{(i)}}(x-\alpha_{2})^{k+\ell-\mu_{1}^{(i)}-j}$ and $$\overline{L}_{\boldsymbol{\mu}^{(i)}}(e_{\ell})=[x^{k}]R(x)=\binom{k+\ell-\mu_{1}^{(i)}}{\ell}(\alpha_{2}-\alpha_{1})^{\ell}.$$
\end{enumerate}
Now we show that $\left(\overline{L}_{\boldsymbol{\mu}^{(i)}}(e_{\ell})\right)_{0\leq i,\ell\leq h_{s,k}-1}$ is nonsingular. By the formula above, we have \begin{align*}
    \det_{0\leq \ell,i\leq h_{s,k}-1}(\overline{L}_{\boldsymbol{\mu}^{(i)}}(e_{\ell}))=(\alpha_{2}-\alpha_{1})^{\frac{(h_{s,k}-1)h_{s,k}}{2}}\det_{0\leq \ell,i\leq h_{s,k}-1}\left(\binom{k+\ell-\mu_{1}^{(i)}}{\ell}\right).
\end{align*}
In fact, we may apply Lemma~\ref{Generalization Vandermonde determinant} over $\mathbb{Q}$ to deduce this: In  Lemma~\ref{Generalization Vandermonde determinant} we define $n=h_{s,k}$, $X_{i}=k-\mu_{1}^{(i-1)}$ and $p_j(x)
=
\frac{1}{(j-1)!}\prod_{t=1}^{j-1}(x+t)$. Then \begin{align*}
    \det_{0\leq \ell,i\leq h_{s,k}-1}\left(\binom{k+\ell-\mu_{1}^{(i)}}{\ell}\right)&=\det_{0\leq i,\ell \leq h_{s,k}-1}\left(\binom{k+\ell-\mu_{1}^{(i)}}{\ell}\right)\\
    &=\det_{1\leq i,j\leq h_{s,k}}(p_{j}(X_{i}))\\
    &=\prod_{j=1}^{h_{s,k}}\frac{1}{(j-1)!}\cdot \prod_{1\leq i<j\leq h_{s,k}}(X_{j}-X_{i})\\
    &=\prod_{j=1}^{h_{s,k}}\frac{1}{(j-1)!}\cdot\prod_{1\leq i<j\leq h_{s,k}}-(j-i) \\
    &=(-1)^{\frac{(h_{s,k}-1)h_{s,k}}{2}}\prod_{j=1}^{h_{s,k}}\frac{1}{(j-1)!}\cdot\prod_{j=1}^{h_{s,k}}(j-1)!\\
    & =(-1)^{\frac{(h_{s,k}-1)h_{s,k}}{2}}\ \text{in}\ \mathbb{Q}.
\end{align*}
Thus $ \det_{0\leq \ell,i\leq h_{s,k}-1}(\overline{L}_{\boldsymbol{\mu}^{(i)}}(e_{\ell}))=(\alpha_{2}-\alpha_{1})^{\frac{(h_{s,k}-1)h_{s,k}}{2}}(-1)^{\frac{(h_{s,k}-1)h_{s,k}}{2}}=(\alpha_{1}-\alpha_{2})^{\frac{(h_{s,k}-1)h_{s,k}}{2}}\neq 0$ in $\mathbb{F}_{q}$, and therefore the matrix $\left(\overline{L}_{\boldsymbol{\mu}^{(i)}}(e_{\ell})\right)_{0\leq i,\ell\leq h_{s,k}-1}$ is nonsingular. \qed

\begin{Prop}\label{r=2}
    If $r=2$, then $$\#\operatorname{DH}(\operatorname{HRS}_k(\alpha,V))= q^{2s-h_{s,k}}(q-1)^{h_{s,k}},$$
    where \begin{align*}
    h_{s,k}=\begin{cases}
            k+2, &1\leq k\leq s-1;\\
            2s-k, &s\leq k \leq 2s-1.
        \end{cases}
    \end{align*}
\end{Prop}

{\it Proof.} To prove Proposition~\ref{r=2}, we shall show the following two claims:\begin{enumerate}
    \item $\{\overline{L}_{\boldsymbol{\mu}}\mid \boldsymbol{\mu}\in \mathcal{M}_{k}\}$ is linearly independent over $\mathbb{F}_{q}.$ \\
    Let $e_{\ell}=(x-\alpha_{1})^{k+\ell}+\mathbb{F}_{q}[x]_{\leq k-1}\in Q_{k}$ for $\ell \in[0,h_{s,k}-1]$. If $\sum_{\boldsymbol{\mu}\in\mathcal{M}_{k}}c_{\boldsymbol{\mu}}\overline{L}_{\boldsymbol{\mu}}=0,$ where $ c_{\boldsymbol{\mu}}\in\mathbb{F}_{q}$, then    $$\sum_{i=0}^{h_{s,k}-1}c_{\boldsymbol{\mu}^{(i)}}\overline{L}_{\boldsymbol{\mu}^{(i)}}(e_{\ell})=0, \quad \ell=0,1,\ldots,h_{s,k}-1,$$
    i.e., $$\begin{pmatrix}
       \overline{L}_{\boldsymbol{\mu}^{(0)}}(e_{0}) & \dots & \overline{L}_{\boldsymbol{\mu}^{(h_{s,k}-1)}}(e_{0})\\
       \vdots & \ddots & \vdots \\
       \overline{L}_{\boldsymbol{\mu}^{(0)}}(e_{h_{s,k}-1}) &\dots & \overline{L}_{\boldsymbol{\mu}^{(h_{s,k}-1)}}(e_{h_{s,k}-1})
    \end{pmatrix}\begin{pmatrix}
        c_{\boldsymbol{\mu}^{(0)}}\\
        \vdots\\
        c_{\boldsymbol{\mu}^{(h_{s,k}-1)}}
    \end{pmatrix}=\begin{pmatrix}
        0\\
        \vdots\\
        0
    \end{pmatrix}.$$
     By Lemma~\ref{Lmu nonsingular},  we have $c_{\boldsymbol{\mu}^{(0)}}=\cdots=c_{\boldsymbol{\mu}^{(h_{s,k}-1)}}=0$ and $\{\overline{L}_{\boldsymbol{\mu}}\mid \boldsymbol{\mu}\in \mathcal{M}_{k}\}$ is linearly independent over $\mathbb{F}_{q}.$
    \item The linear map \begin{align*}
         \Phi:Q_{k}\longrightarrow\mathbb{F}_{q}^{h_{s,k}}
     \end{align*}
     defined by $\overline{f}\mapsto(\overline{L}_{\boldsymbol{\mu}}(\overline{f}))_{\boldsymbol{\mu}\in\mathcal{M}_{k}}$ is surjective.\\
     Indeed, suppose that $ \Img  \Phi\subsetneq \mathbb{F}_{q}^{h_{s,k}}$. Then there exists a nonzero linear map $\phi:\mathbb{F}_{q}^{h_{s,k}}\rightarrow\mathbb{F}_{q}$ such that $\phi(\mathbf{y})=0$ for all $\mathbf{y}\in\Img\Phi$. We may write $\phi((a_{\boldsymbol{\mu}})_{\boldsymbol{\mu}\in\mathcal{M}_{k}})=\sum_{\boldsymbol{\mu}\in\mathcal{M}_{k}}c_{\boldsymbol{\mu}}a_{\boldsymbol{\mu}}$ for some $\mathbf{0}\neq(c_{\boldsymbol{\mu}})_{\boldsymbol{\mu}\in\mathcal{M}_{k}}\in\mathbb{F}_{q}^{h_{s,k}}.$ Then $0=\phi(\Phi(\overline{f}))=\sum_{\boldsymbol{\mu}\in\mathcal{M}_{k}}c_{\boldsymbol{\mu}}\overline{L}_{\boldsymbol{\mu}}(\overline{f})$ for all $\overline{f}\in Q_{k}.$ This implies $\sum_{\boldsymbol{\mu}\in\mathcal{M}_{k}}c_{\boldsymbol{\mu}}\overline{L}_{\boldsymbol{\mu}}=0$, a contradiction to the linear independence established above.
\end{enumerate} Thus  $$\dim_{\mathbb{F}_{q}}(\Phi^{-1}(\mathbf{y}))=\dim_{\mathbb{F}_{q}}\Ker\Phi=\dim_{\mathbb{F}_{q}}Q_{k}-h_{s,k}=2s-k-h_{s,k}\quad \text{for}\ \mathbf{y}\in\mathbb{F}_{q}^{h_{s,k}}.$$
     It follows from the proof of Corollary~\ref{Count Deep holes} that
     \begin{align*}
         \# \operatorname{DH}(\operatorname{HRS}_{k}(\alpha,V))
 =q^{k}\cdot \#\left(\Phi^{-1}((\mathbb{F}_{q}^{*})^{h_{s,k}})\right) =q^{k}\cdot (q-1)^{h_{s,k}} \cdot q^{2s-k-h_{s,k}} =q^{2s-h_{s,k}}(q-1)^{h_{s,k}}.
     \end{align*}\qed

\begin{Prop}\label{k=1}
   Suppose $k=1$. Then:
   \begin{enumerate}
        \item If $s=1$, then $$\operatorname{DH}(\operatorname{HRS}_1(\alpha,V))=\{U\in \mathbb{F}_{q}^{1\times r}\mid f_{U}(\alpha_{i})\neq f_{U}(\alpha_{j}) \ \text{for all} \ i\neq j\}$$
        and $\#\operatorname{DH}(\operatorname{HRS}_1(\alpha,V))= \displaystyle\frac{q!}{(q-r)!}.$
        \item If $s\geq 2$, then \[
        \operatorname{DH}(\operatorname{HRS}_1(\alpha,V))
        =
        \left\{
        U\in \mathbb{F}_{q}^{s\times r}\middle|\;
\begin{array}{l}
f_U(\alpha_i)\neq f_U(\alpha_j)\text{ for all }i\neq j,\\[1mm]
\partial^{(1)}f_U(\alpha_j)\neq0 \text{ for every }j\in[r]
\end{array}
\right\},
\]
and
\[
\#\operatorname{DH}(\operatorname{HRS}_1(\alpha,V))
=\frac{q!}{(q-r)!}(q-1)^r q^{r(s-2)}.
\]
       \end{enumerate}
\end{Prop}

{\it Proof.} By Theorem~\ref{Deep holes}, $U\in\operatorname{DH}(\operatorname{HRS}_1(\alpha,V))$ if and only if $\deg(f_{U}\bmod M_{\boldsymbol{\mu}})=1$ for all $ \ \boldsymbol{\mu}\in\mathcal{M}_{1}.$ By the definition of $\mathcal{M}_{1}$, we have:
    \begin{enumerate}
        \item If $\boldsymbol{\mu}\in\mathcal{M}_{1}$ with $\mu_{j_{0}}=\mu_{j_{1}}=1$ for some $j_{0}\neq j_{1}$, then we have $M_{\boldsymbol{\mu}}(x)=(x-\alpha_{j_{0}})(x-\alpha_{j_{1}}).$ Then by the definition of $f_{U}\bmod M_{\boldsymbol{\mu}}$ and the Hasse-Taylor expansion of $f_{U}\bmod M_{\boldsymbol{\mu}}$ at $\alpha_{j_{0}}$, we have $f_{U}\bmod M_{\boldsymbol{\mu}}= f_{U}(\alpha_{j_{0}})+\partial^{(1)}(f_{U}\bmod M_{\boldsymbol{\mu}})(\alpha_{j_{0}})(x-\alpha_{j_{0}})$ with $(f_{U}\bmod M_{\boldsymbol{\mu}})(\alpha_{j_{1}})=f_{U}(\alpha_{j_{1}})$.   Thus  $\deg(f_{U}\bmod M_{\boldsymbol{\mu}})=1$ if and only if $$f_{U}(\alpha_{j_{0}})\neq f_{U}(\alpha_{j_{1}}).$$
        \item On the other hand, if $\boldsymbol{\mu}\in\mathcal{M}_{1}$ with $\mu_{j_{2}}=2$ for some $j_{2}\in[r]$, then we have $M_{\boldsymbol{\mu}}(x)=(x-\alpha_{j_{2}})^{2}$ and $f_{U}\bmod M_{\boldsymbol{\mu}}=f_{U}(\alpha_{j_{2}})+\partial^{(1)}f_{U}(\alpha_{j_{2}})(x-\alpha_{j_{2}})$ by the definition of $f_{U}\bmod M_{\boldsymbol{\mu}}$ and the Hasse-Taylor expansion of $f_{U}\bmod M_{\boldsymbol{\mu}}$ at $\alpha_{j_{2}}$. 
        Hence $\deg(f_{U}\bmod M_{\boldsymbol{\mu}})=1$ if and only if $$\partial^{(1)}f_{U}(\alpha_{j_{2}})\neq 0.
        $$
    \end{enumerate}
    
    Thus, if $s=1$, then  only the first case occurs. So $$\operatorname{DH}(\operatorname{HRS}_1(\alpha,V))=\{U\in \mathbb{F}_{q}^{1\times r}\mid f_{U}(\alpha_{i})\neq f_{U}(\alpha_{j}) \ \text{for all} \ i\neq j\}$$ and $$\#\operatorname{DH}(\operatorname{HRS}_1(\alpha,V))=q(q-1)\cdots(q-r+1)=\frac{q!}{(q-r)!}.$$

    If $s\geq 2$, then both cases occur and \begin{align*}
\operatorname{DH}\bigl(\operatorname{HRS}_1(\alpha,V)\bigr)
=
\left\{
    U\in\mathbb{F}_{q}^{s\times r}
    \;\middle|\;
      \begin{array}{l}
f_U(\alpha_i)\neq f_U(\alpha_j)\text{ for all }i\neq j,\\[1mm]
\partial^{(1)}f_U(\alpha_j)\neq0 \text{ for every }j\in[r]
\end{array}
\right\}.
\end{align*}
        Thus there are $q(q-1)\cdots(q-r+1)=\frac{q!}{(q-r)!}$ choices for $(f_{U}(\alpha_{j}))_{j\in[r]}$ and $(q-1)^{r}$ choices for $(\partial^{(1)}f_U(\alpha_j))_{j\in[r]}$. The values $\partial^{(i)}f_{U}(\alpha_{j})$ may be chosen independently from $\mathbb{F}_{q}$ for $i\in[2,s-1]$ and $j\in[r]$ and hence contribute $q^{r(s-2)}$ possibilities. Therefore
        $$\#\operatorname{DH}(\operatorname{HRS}_1(\alpha,V))= \displaystyle\frac{q!}{(q-r)!}(q-1)^{r}q^{r(s-2)}.$$
        This finishes the proof.\qed

\medskip

The following lemma is  Segre's theorem.
\begin{Lem}\cite[Theorem I]{Seg55}\label{Seger theorem}
    If $\text{char}(\mathbb{F}_{q})$ is odd, every oval in $\operatorname{PG}(2,q)$ is a conic (i.e., can be represented by an equation of degree $2$).
\end{Lem}

\begin{Prop}
    Suppose $k=2$.  Let
    \[
    A_\alpha(r)
    =
    \left\{(y_1,\dots,y_r)\in\mathbb F_q^r\ \middle|\ \frac{y_{j}-y_{i}}{\alpha_{j}-\alpha_{i}}\neq\frac{y_{\ell}-y_{i}}{\alpha_{\ell}-\alpha_{i}}\ \text{for all} \ 1\leq i<j<\ell\leq r
    \right\}.
    \] Then the following statements hold.
    \begin{enumerate}
        \item If $s=1$, then \begin{align*}
\operatorname{DH}\bigl(\operatorname{HRS}_2(\alpha,V)\bigr)
=
\left\{
    U\in\mathbb{F}_{q}^{1\times r}
    \;\middle|\;
        \frac{f_U(\alpha_j)-f_U(\alpha_i)}
              {\alpha_j-\alpha_i}
        \neq
        \frac{f_U(\alpha_\ell)-f_U(\alpha_i)}
              {\alpha_\ell-\alpha_i}
        \text{ for all }1\leq i<j<\ell\leq r
\right\},
\end{align*}
        and $\#\operatorname{DH}(\operatorname{HRS}_2(\alpha,V))= \#A_\alpha(r).$
        \item If $s=2,$ then \begin{align*}
\operatorname{DH}\bigl(\operatorname{HRS}_2(\alpha,V)\bigr)
=
\left\{
    U\in\mathbb{F}_{q}^{2\times r}
    \;\middle|\;
    \begin{array}{@{}l@{}}
        \dfrac{f_U(\alpha_j)-f_U(\alpha_i)}
              {\alpha_j-\alpha_i}
        \neq
        \dfrac{f_U(\alpha_\ell)-f_U(\alpha_i)}
              {\alpha_\ell-\alpha_i}
       \text{ for all }1\leq i<j<\ell\leq r
        \\[6pt]
        f_U(\alpha_j)
        \neq f_U(\alpha_i)
        +\partial^{(1)}f_U(\alpha_i)(\alpha_j-\alpha_i)
         \text{ for all }i,j\in[r],\ i\neq j
    \end{array}
\right\},
\end{align*}
        and $\#\operatorname{DH}(\operatorname{HRS}_2(\alpha,V))= \#A_\alpha(r)(q-r+1)^{r}.$
        \item If $s\geq 3$, then \begin{align*}
\operatorname{DH}\bigl(\operatorname{HRS}_2(\alpha,V)\bigr)=
\left\{
    U\in\mathbb{F}_{q}^{s\times r}
    \;\middle|\;
    \begin{array}{@{}l@{}}
        \dfrac{f_U(\alpha_j)-f_U(\alpha_i)}
              {\alpha_j-\alpha_i}
        \neq
        \dfrac{f_U(\alpha_\ell)-f_U(\alpha_i)}
              {\alpha_\ell-\alpha_i} \text{ for all }1\leq i<j<\ell\leq r
        \\[6pt]
        \partial^{(1)}f_U(\alpha_i)
        \neq
        \dfrac{f_U(\alpha_j)-f_U(\alpha_i)}
              {\alpha_j-\alpha_i} \text{ for all }i,j\in[r],\ i\neq j;
        \\[6pt]
        \partial^{(2)}f_U(\alpha_j)\neq 0\text{ for all }j\in[r]
    \end{array}
\right\},
\end{align*}
        and $\#\operatorname{DH}(\operatorname{HRS}_2(\alpha,V))= \#A_\alpha(r)(q-r+1)^{r}(q-1)^{r}q^{r(s-3)}.$
    \end{enumerate}
    In particular, $\#A_{\alpha}(1)=q$, $\#  A_{\alpha}(2)=q^{2}$, $\#A_{\alpha}(3)=q^{2}(q-1)$, $\#A_{\alpha}(4)=q^{2}(q-1)(q-3)$
    and $$\#A_{\mathbb{F}_{q}}(q)=q^{2}(q-1) \quad\text{for} \  q\ \text{odd}.$$
\end{Prop}

{\it Proof.} By Theorem~\ref{Deep holes}, $U\in\operatorname{DH}(\operatorname{HRS}_2(\alpha,V))$ if and only if $\deg(f_{U}\bmod M_{\boldsymbol{\mu}})=2$ for all $\boldsymbol{\mu}\in\mathcal{M}_{2}.$ By the definition of $\mathcal{M}_{2}$, we have:
    \begin{enumerate}
        \item If $\boldsymbol{\mu}\in\mathcal{M}_{2}$ with $\mu_{j_{0}}=\mu_{j_{1}}=\mu_{j_{2}}=1$ for some pairwise distinct $j_{0},j_{1},j_{2}$, then $M_{\boldsymbol{\mu}}(x)=(x-\alpha_{j_{0}})(x-\alpha_{j_{1}})(x-\alpha_{j_{2}}).$ 
        Thus $\deg(f_{U}\bmod M_{\boldsymbol{\mu}})=2$ if and only if the three points $(\alpha_{j},f_{U}(\alpha_{j}))$ for $j=j_{0},j_{1},j_{2}$ are not collinear, equivalently,$$ \frac{f_{U}(\alpha_{j_{1}})-f_{U}(\alpha_{j_{0}})}{\alpha_{j_{1}}-\alpha_{j_{0}}}\neq\frac{f_{U}(\alpha_{j_{2}})-f_{U}(\alpha_{j_{0}})}{\alpha_{j_{2}}-\alpha_{j_{0}}}.$$
        \item Second, if $\boldsymbol{\mu}\in\mathcal{M}_{2}$ with $\mu_{j_{3}}=2$ and $\mu_{j_{4}}=1$, then  $M_{\boldsymbol{\mu}}=(x-\alpha_{j_{3}})^{2}(x-\alpha_{j_{4}}).$ By the definition of $f_{U}\bmod M_{\boldsymbol{\mu}}$, the Hasse-Taylor expansion of $f_{U}\bmod M_{\boldsymbol{\mu}}$ at $\alpha_{j_{3}}$ gives $$f_{U}\bmod M_{\boldsymbol{\mu}}=f_{U}(\alpha_{j_{3}})+\partial^{(1)}f_{U}(\alpha_{j_{3}})(x-\alpha_{j_{3}})+\partial^{(2)}(f_{U}\bmod M_{\boldsymbol{\mu}})(\alpha_{j_{3}})(x-\alpha_{j_{3}})^{2}$$ with $(f_{U}\bmod M_{\boldsymbol{\mu}})(\alpha_{j_{4}})=f_{U}(\alpha_{j_{4}})$. 
        This implies that $\deg(f_{U}\bmod M_{\boldsymbol{\mu}})\leq 1$ if and only if $$f_{U}(\alpha_{j_{4}})=f_{U}(\alpha_{j_{3}})+\partial^{(1)}f_{U}(\alpha_{j_{3}})(\alpha_{j_{4}}-\alpha_{j_{3}}).$$ Thus $\deg(f_{U}\bmod M_{\boldsymbol{\mu}})=2$ if and only if $$\partial^{(1)}f_{U}(\alpha_{j_{3}})\neq \frac{f_{U}(\alpha_{j_{4}})-f_{U}(\alpha_{j_{3}})}{\alpha_{j_{4}}-\alpha_{j_{3}}}.$$
        \item Third, if $\boldsymbol{\mu}\in\mathcal{M}_{2}$ with $\mu_{j_{5}}=3$, then $M_{\boldsymbol{\mu}}(x)=(x-\alpha_{j_{5}})^{3}$ and $f_{U}\bmod M_{\boldsymbol{\mu}}=f_{U}(\alpha_{j_{5}})+\partial^{(1)}f_{U}(\alpha_{j_{5}})(x-\alpha_{j_{5}})+\partial^{(2)}f_{U}(\alpha_{j_{5}})(x-\alpha_{j_{5}})^{2}$ by the definition of $f_{U}\bmod M_{\boldsymbol{\mu}}$ and the Hasse-Taylor expansion of $f_{U}\bmod M_{\boldsymbol{\mu}}$ at $\alpha_{j_{5}}$. Hence $\deg(f_{U}\bmod M_{\boldsymbol{\mu}})=2$ if and only if $$\partial^{(2)}f_{U}(\alpha_{j_{5}})\neq 0.$$
    \end{enumerate}

    Thus, if $s=1$, then only the first type occurs. So \begin{align*}
\operatorname{DH}\bigl(\operatorname{HRS}_2(\alpha,V)\bigr)
=
\left\{
    U\in\mathbb{F}_{q}^{1\times r}
    \;\middle|\;
        \frac{f_U(\alpha_j)-f_U(\alpha_i)}
              {\alpha_j-\alpha_i}
        \neq
        \frac{f_U(\alpha_\ell)-f_U(\alpha_i)}
              {\alpha_\ell-\alpha_i}\text{ for all }1\leq i<j<\ell\leq r
\right\},
\end{align*} and $\#\operatorname{DH}(\operatorname{HRS}_2(\alpha,V))=\#A_{\alpha}(r)$ by the definition of $A_{\alpha}(r)$.

    If $s=2$, then the first two cases occur and \begin{align*}
\operatorname{DH}\bigl(\operatorname{HRS}_2(\alpha,V)\bigr)
=
\left\{
    U\in\mathbb{F}_{q}^{2\times r}
    \;\middle|\;
    \begin{array}{@{}l@{}}
        \dfrac{f_U(\alpha_j)-f_U(\alpha_i)}
              {\alpha_j-\alpha_i}
        \neq
        \dfrac{f_U(\alpha_\ell)-f_U(\alpha_i)}
              {\alpha_\ell-\alpha_i} \text{ for all }1\leq i<j<\ell\leq r
        \\[6pt]
        f_U(\alpha_j)
        \neq f_U(\alpha_i)
        +\partial^{(1)}f_U(\alpha_i)(\alpha_j-\alpha_i)  \text{ for all }i,j\in[r],\ i\neq j
    \end{array}
\right\}.
\end{align*}
   Clearly, there are $\#A_{\alpha}(r)$ choices for $(f_{U}(\alpha_{j}))_{j\in[r]}$.  For any $i\in[r]$, there are $q-(r-1)=q-r+1$ choices for $\partial^{(1)}f_{U}(\alpha_{i})$ by the first condition. Thus $$\#\operatorname{DH}(\operatorname{HRS}_2(\alpha,V))=\#A_{\alpha}(r)(q-r+1)^{r}.$$

    Finally, if $s\geq 3$, then all three cases occur and
   \begin{align*}
\operatorname{DH}\bigl(\operatorname{HRS}_2(\alpha,V)\bigr)=
\left\{
    U\in\mathbb{F}_{q}^{s\times r}
    \;\middle|\;
    \begin{array}{@{}l@{}}
        \dfrac{f_U(\alpha_j)-f_U(\alpha_i)}
              {\alpha_j-\alpha_i}
        \neq
        \dfrac{f_U(\alpha_\ell)-f_U(\alpha_i)}
              {\alpha_\ell-\alpha_i} \text{ for all }1\leq i<j<\ell\leq r
        \\[6pt]
        \partial^{(1)}f_U(\alpha_i)
        \neq
        \dfrac{f_U(\alpha_j)-f_U(\alpha_i)}
              {\alpha_j-\alpha_i}  \text{ for all }i,j\in[r],\ i\neq j
        \\[6pt]
        \partial^{(2)}f_U(\alpha_j)\neq 0 \text{ for all }j\in[r]
    \end{array}
\right\}.
\end{align*}
    Thus there are $\#A_{\alpha}(r)$ choices for $(f_{U}(\alpha_{j}))_{j\in[r]}$, $(q-r+1)^{r}$ choices for $(\partial^{(1)}f_U(\alpha_j))_{j\in[r]}$ and $(q-1)^{r}$ choices for $(\partial^{(2)}f_{U}(\alpha_{j}))_{j\in[r]}$. The values $\partial^{(i)}f_{U}(\alpha_{j})$ may be chosen independently from $\mathbb{F}_{q}$ for $i\in[3,s-1]$ and $j\in[r]$ and hence contribute $q^{r(s-3)}$ possibilities. Therefore $$\#\operatorname{DH}(\operatorname{HRS}_2(\alpha,V))=\#A_{\alpha}(r)(q-r+1)^{r}(q-1)^{r}q^{r(s-3)}.$$

Next, we consider several cases of $A_{\alpha}(r)$.

    Clearly, $\#A_{\alpha}(1)=q$ and $ \#A_{\alpha}(2)=q^{2}.$ If $r=3$, then \begin{align*}
        \#A_{\alpha}(3)&=q^{3}-\#\{(y_{1},y_{2},y_{3})\in\mathbb{F}_{q}^{3}\mid (\alpha_{3}-\alpha_{2})y_{1}-(\alpha_{3}-\alpha_{1})y_{2}+(\alpha_{2}-\alpha_{1})y_{3}=0\}\\
    &=q^{3}-q^{2}=q^{2}(q-1). \tag*{$\alpha_{1}$ and $\alpha_{2},\alpha_{3}$ are pairwise distinct} \end{align*}
    If $r=4$, then \begin{align*}
        A_{\alpha}(4)=\{(y_{1},y_{2},y_{3},y_{4})\in\mathbb{F}_{q}^{4}\mid(y_{1},y_{2},y_{3})\in A_{\alpha}(3), y_{4}\neq y_{i}+(y_{j}-y_{i})\frac{\alpha_{4}-\alpha_{i}}{\alpha_{j}-\alpha_{i}}\ \text{for all}\ 1\leq i<j\leq 3\}.
    \end{align*} 
    We claim that if $ y_{i}+(y_{j}-y_{i})\frac{\alpha_{4}-\alpha_{i}}{\alpha_{j}-\alpha_{i}}= y_{i'}+(y_{j'}-y_{i'})\frac{\alpha_{4}-\alpha_{i'}}{\alpha_{j'}-\alpha_{i'}}$, then $(i,j)=(i',j').$ Note that $y_{i}+(y_{j}-y_{i})\frac{\alpha_{4}-\alpha_{i}}{\alpha_{j}-\alpha_{i}}=y_{j}+(y_{i}-y_{j})\frac{\alpha_{4}-\alpha_{j}}{\alpha_{i}-\alpha_{j}}$. Since $1\leq i<j\leq 3$ and $1\leq i'<j'\leq 3$, there exists $a\in\{i,j\}\cap \{i',j'\}$. So we may write $\{a,b\}=\{i,j\}$ and $\{a,c\}=\{i',j'\}.$ Then we have $$y_{a}+(y_{b}-y_{a})\frac{\alpha_{4}-\alpha_{a}}{\alpha_{b}-\alpha_{a}}=y_{a}+(y_{c}-y_{a})\frac{\alpha_{4}-\alpha_{a}}{\alpha_{c}-\alpha_{a}}.$$ This implies $\frac{y_{b}-y_{a}}{\alpha_{b}-\alpha_{a}}=\frac{y_{c}-y_{a}}{\alpha_{c}-\alpha_{a}}.$ By the definition of $A_{\alpha}(r)$, we have $b=c$ and $(i,j)=(i',j')$. Thus for any $(y_{1},y_{2},y_{3})\in A_{\alpha}(3)$, there are $q-3$ choices for $y_{4}$ and $\#A_{\alpha}(4)=q^{2}(q-1)(q-3).$
Finally, If $r=q$ then \begin{align*}
\#A_{\mathbb{F}_{q}}(q)
=
\#\left\{
    f:\mathbb{F}_{q}\rightarrow\mathbb{F}_{q}
  \mid 
        \Gamma_f
        :=\{(x,f(x))\mid x\in\mathbb{F}_{q}\} \text{ contains no three collinear points}
\right\}.
\end{align*}
We claim that \begin{align*}
    &\{f:\mathbb{F}_{q}\rightarrow\mathbb{F}_{q}\mid \Gamma_{f}:=\{(x,f(x))\mid x\in\mathbb{F}_{q}\}\ \text{contains no three collinear points}\}\\
    =&\{f(x)=ax^{2}+bx+c\mid a\in\mathbb{F}_{q}^{*},b,c\in\mathbb{F}_{q}\}.
\end{align*}
If $\Gamma_{f}:=\{(x,f(x))\mid x\in\mathbb{F}_{q}\}$ contains no three collinear points, then $\{[x,f(x),1]\mid x\in\mathbb{F}_{q}\}\cup\{[0,1,0]\}$ is an oval in $\operatorname{PG}(2,q)$.  Since $q$ is odd, Lemma~\ref{Seger theorem} implies that 
there exist $A,B,C,D,E,F\in\mathbb{F}_{q}$, not all zero, such that \begin{align*}
\left\{
    [X:Y:Z]
    \;\middle|\;
    AX^{2}+BXZ+CZ^{2}+DYZ+EXY+FY^{2}=0
\right\}
=
\left\{
    [x:f(x):1]
    \;\middle|\;
    x\in\mathbb{F}_{q}
\right\}
\cup \bigl\{[0:1:0]\bigr\}.
\end{align*} Since $[0,1,0]$ lies on this conic, we have $F=0$. Moreover, since $[0,1,0]$ is the unique point of the conic on the line $Z=0$, while $[E,-A,0]$ also lies on the conic, it follows that $E=0$. Therefore we obtain $Ax^{2}+Bx+C+Df(x)=0$ for all $x\in\mathbb{F}_{q}.$ We suppose $D=0$. It follows from $Ax^{2}+Bx+C=0$ for all $x\in\mathbb{F}_{q}$ and $q\geq 3$ that $A=B=C=0$, a contradiction to the assumption of $A,B,C,D,E,F$. Hence $f(x)=-\frac{A}{D}x^{2}-\frac{B}{D}x-\frac{C}{D}$. Suppose that $A=0$, then we have $\begin{vmatrix}
    1 & 0 & -\frac{C}{D}\\
    1 & 1 & -\frac{B}{D}-\frac{C}{D}\\
    1 & -1 & \frac{B}{D}-\frac{C}{D}
\end{vmatrix}=0$. This implies that $(0,f(0)),(1,f(1))$ and $(-1,f(-1))$ are collinear, a contradiction to the assumption of $\Gamma_{f}$. Thus $f(x)=ax^{2}+bx+c$ for some $a\in\mathbb{F}_{q}^{*},b,c\in\mathbb{F}_{q}$. Conversely, if $f=ax^{2}+bx+c$ for $a\in\mathbb{F}_{q}^{*},b,c\in\mathbb{F}_{q}$ then $\Gamma_{f}$ has no three collinear points. So our claim
follows.  Therefore $\#A_{\mathbb{F}_{q}}(q)=q^{2}(q-1)$ for $q$ odd\qed

\begin{Prop}\label{k=sr-1}
    If $k=sr-1$, then $$\operatorname{DH}(\operatorname{HRS}_{sr-1}(\alpha,V))=\mathbb{F}_{q}^{s\times r}\backslash \operatorname{HRS}_{sr-1}(\alpha,V)$$ and $\#\operatorname{DH}(\operatorname{HRS}_{sr-1}(\alpha,V))=(q-1)q^{sr-1}.$
\end{Prop}

{\it Proof.} Note that under the assumption, $\mathcal{M}_{sr-1}=\{(s,s,\dots,s)\}$. By Theorem~\ref{Deep holes}, $U\in\operatorname{DH}(\operatorname{HRS}_{sr-1}(\alpha,V))$ if and only if $\deg(f_{U}\bmod \pi_{\alpha}(x)^{s})=sr-1.$  Since $f_{U}\in\mathbb{F}_{q}[x]_{\leq sr-1}$, we have $f_{U}\bmod M_{\boldsymbol{\mu}}=f_{U}$. This implies that $U\in \operatorname{DH}(\operatorname{HRS}_{sr-1} (\alpha,V))$ if and only if $\deg(f_{U})=sr-1$. Thus $\operatorname{DH}(\operatorname{HRS}_{sr-1}(\alpha,V))=\mathbb{F}_{q}^{s\times r}\backslash \operatorname{HRS}_{sr-1}(\alpha,V),$ and $\#\operatorname{DH}(\operatorname{HRS}_{sr-1}(\alpha,V))=q^{sr}-q^{sr-1}=(q-1)q^{sr-1}.$\qed

\begin{Prop}
    If $k=sr-2$, then\begin{align*}
        \operatorname{DH}(\operatorname{HRS}_{sr-2}(\alpha,V))=\{&U\in\mathbb{F}_{q}^{s\times r}\mid f_{U}(x)=ax^{sr-1}+bx^{sr-2}+g(x) \ \text{with} \\  
        &  g(x)\in\mathbb{F}_{q}[x]_{\leq sr-3} \ \text{and} \ a(\alpha_{j}-s\sum_{\ell=1}^{r}\alpha_{\ell})\neq b\ \text{for all} \ j\in [r]\}
    \end{align*}
    and $\#\operatorname{DH}(\operatorname{HRS}_{sr-2}(\alpha,V))=(q-1)(q-r+1)q^{sr-2}.$
\end{Prop}

{\it Proof.} By Theorem~\ref{Deep holes}, $U\in\operatorname{DH}(\operatorname{HRS}_{sr-2}(\alpha,V))$ if and only if $\deg(f_{U}\bmod M_{\boldsymbol{\mu}})=sr-2$ for all $ \boldsymbol{\mu}\in\mathcal{M}_{sr-2}.$ By the definition of $\mathcal{M}_{sr-2}$, we have $\mathcal{M}_{sr-2}=\{\boldsymbol{\mu}^{(j)}\mid j\in [r]\}$ where $\boldsymbol{\mu}^{(j)}$ denotes the vector whose $j$-th entry is $s-1$ and all other entries are $s$. Hence $M_{\boldsymbol{\mu}^{(j)}}(x)=(x-\alpha_{j})^{s-1}\prod_{\ell\neq j}(x-\alpha_{\ell})^{s}$ with $[x^{sr-2}]M_{\boldsymbol{\mu}^{(j)}}(x)=-((s-1)\alpha_{j}+s\sum_{\ell\neq j}\alpha_{\ell})=\alpha_{j}-s\sum_{\ell=1}^{r}\alpha_{\ell}.$ Now we may write $f_{U}(x)=ax^{sr-1}+bx^{sr-2}+g(x)$ for some $g(x)\in\mathbb{F}_{q}[x]_{\leq sr-3}$. This implies $f_{U}(x)\bmod M_{\boldsymbol{\mu}^{(j)}}(x)=f_{U}(x)-aM_{\boldsymbol{\mu}^{(j)}}(x)$ with $[x^{sr-2}](f_{U}(x)\bmod M_{\boldsymbol{\mu}^{(j)}}(x))=b-a(\alpha_{j}-s\sum_{\ell=1}^{r}\alpha_{\ell}).$ Thus \begin{align*}
        \operatorname{DH}(\operatorname{HRS}_{sr-2}(\alpha,V))=\{&U\in\mathbb{F}_{q}^{s\times r}\mid f_{U}(x)=ax^{sr-1}+bx^{sr-2}+g(x) \ \text{with} \\  
        &  g(x)\in\mathbb{F}_{q}[x]_{\leq sr-3} \ \text{and} \ a(\alpha_{j}-s\sum_{\ell=1}^{r}\alpha_{\ell})\neq b\ \text{for all} \ j\in [r]\}.
    \end{align*}
    Thus there are $q^{sr-2}$ choices for $g(x)\in\mathbb{F}_{q}[x]_{\leq sr-3}$.  If $a=0,$ then there are $q-1$ choices for $b\in\mathbb{F}_{q}^{*}$. If $a\neq 0$, then there are $q-r$ choices for $b$ since the values $ a(\alpha_{j}-s\sum_{\ell=1}^{r}\alpha_{\ell})$ for $j\in[r]$ are pairwise distinct. Thus there are $(q-1)+(q-1)(q-r)=(q-1)(q-r+1)$ choices for $(a,b)$ and $$\#\operatorname{DH}(\operatorname{HRS}_{sr-2}(\alpha,V))=(q-1)(q-r+1)q^{sr-2}.$$ \qed

    \begin{Prop}
        If $k=sr-3$, then \begin{align*}
\operatorname{DH}\bigl(
    \operatorname{HRS}_{sr-3}(\alpha,V)
\bigr)=
\left\{
    U\in\mathbb{F}_{q}^{s\times r}
    \;\middle|\;
    \begin{array}{@{}l@{}}
        f_U(x)
        =ax^{sr-1}+bx^{sr-2}+cx^{sr-3}+g(x),
        g\in\mathbb{F}_{q}[x]_{\leq sr-4},
        \\[6pt]
        c+Tb+(T^{2}-E)a-u(b+Ta)+va\neq 0 
        \text{ for all }(u,v)\in\mathcal{P}_{\alpha}^{(s)}
    \end{array}
\right\},
\end{align*}
    where $T=s\sum_{j=1}^{r}\alpha_{j}$, $E=[x^{sr-2}]\pi_{\alpha}(x)^{s}$ and
    \[
    \mathcal P_\alpha^{(s)}
    =
    \begin{cases}
    \left\{
    (\alpha_i+\alpha_j,\alpha_i\alpha_j)
    :
    1\le i<j\le r
    \right\}, & s=1;\\
    \left\{
    (\alpha_i+\alpha_j,\alpha_i\alpha_j)
    :
    1\le i\le j\le r
    \right\}, & s\ge2.
    \end{cases}
    \]
    and
    \[
    \#\operatorname{DH}(\operatorname{HRS}_{sr-3}(\alpha,V))
    =
    q^{sr-3}(q-1)
    \left(
    L_0(\mathcal P_\alpha^{(s)})+1
    \right),
    \]
    where \(L_0(\mathcal P_\alpha^{(s)})\) is the number of affine lines in
    \(\mathbb F_q^2\) avoiding \(\mathcal P_\alpha^{(s)}\). Moreover,
    \[
    L_0(\mathcal P_\alpha^{(s)})
    =
    \sum_{\lambda\in\mathbb F_q}
    \prod_{(u,v)\in\mathcal P_\alpha^{(s)}}
    (u-\lambda)^{q-1} +\sum_{\beta,\gamma\in\mathbb F_q}
    \prod_{(u,v)\in\mathcal P_\alpha^{(s)}}
    (v-\beta u-\gamma)^{q-1}.
    \]
    \end{Prop}

{\it Proof.} By Theorem~\ref{Deep holes}, $U\in\operatorname{DH}(\operatorname{HRS}_{sr-3}(\alpha,V))$ if and only if $\deg(f_{U}\bmod M_{\boldsymbol{\mu}})=sr-3$ for all $\boldsymbol{\mu}\in\mathcal{M}_{sr-3}.$  Then $$\{M_{\boldsymbol{\mu}}(x)\mid \boldsymbol{\mu}\in\mathcal{M}_{sr-3}\}=\left\{\frac{\pi_{\alpha}(x)^{s}}{Q(x)}\ \middle| \ Q(x)=x^{2}-ux+v \ \text{with}\ (u,v)\in\mathcal{P}_{\alpha}^{(s)}\right\}.$$
    By comparison of the coefficients of $x^{sr-1}$ and $x^{sr-2}$ in $\pi_{\alpha}(x)^{s}=(x^{2}-ux+v)\frac{\pi_{\alpha}(x)^{s}}{Q(x)}$, we have that  $$[x^{sr-3}](\frac{\pi_{\alpha}(x)^{s}}{Q(x)})=u-T\quad \text{and} \quad [x^{sr-4}](\frac{\pi_{\alpha}(x)^{s}}{Q(x)})=E+u^{2}-uT-v.$$  If $f_{U}(x)=ax^{sr-1}+bx^{sr-2}+cx^{sr-3}+g(x)$ for some $g(x)\in\mathbb{F}_{q}[x]_{\leq sr-4}$, then $f_{U}(x)\bmod \frac{\pi_{\alpha}(x)^{s}}{Q(x)}=f_{U}(x)-(ax+b-a(u-T))\frac{\pi_{\alpha}(x)^{s}}{Q(x)}$. This implies  $[x^{sr-3}](f_{U}(x)\bmod \frac{\pi_{\alpha}(x)^{s}}{Q(x)})=c+Tb+(T^{2}-E)a-u(b+Ta)+va$. Thus\begin{align*}
\operatorname{DH}\bigl(
    \operatorname{HRS}_{sr-3}(\alpha,V)
\bigr)=
\left\{
    U\in\mathbb{F}_{q}^{s\times r}
    \;\middle|\;
    \begin{array}{@{}l@{}}
        f_U(x)
        =ax^{sr-1}+bx^{sr-2}+cx^{sr-3}+g(x),
        g\in\mathbb{F}_{q}[x]_{\leq sr-4},
        \\[6pt]
        c+Tb+(T^{2}-E)a-u(b+Ta)+va\neq 0  \text{ for all }(u,v)\in\mathcal{P}_{\alpha}^{(s)}
    \end{array}
\right\}.
\end{align*}
    There are $q^{sr-3}$ choices for $g\in\mathbb{F}_{q}[x]_{\leq sr-4}$. Set $A=a,B=b+Ta$ and $C=c+Tb+(T^{2}-E)a$. Note that the map $\varphi:\mathbb{F}_{q}^{3}\rightarrow\mathbb{F}_{q}^{3}$  defined by $(a,b,c)\mapsto (A,B,C)$ is an $\mathbb{F}_{q}$-linear  isomorphism. If $(A,B)=(0,0)$, then there are $q-1$ choices for $C\in\mathbb{F}_{q}^{*}.$ If $(A,B)\neq (0,0),$ then the affine line $C-XB+YA=0$  avoids $\mathcal{P}_{\alpha}^{(s)}.$ Conversely, every affine line avoiding $\mathcal{P}_{\alpha}^{(s)}$ is represented by exactly $q-1$ nonzero scalar multiples of such triples $(A,B,C).$ Therefore there are $(q-1)(L_{0}(\mathcal{P}_{\alpha}^{(s)})+1)$ choices for$(A,B,C).$ Hence
    \[
    \#\operatorname{DH}(\operatorname{HRS}_{sr-3}(\alpha,V))
    =
    q^{sr-3}(q-1)
    \left(
    L_0(\mathcal P_\alpha^{(s)})+1
    \right).
    \]
 By the definition of $L_0(\mathcal P_\alpha^{(s)})$, we have \begin{align*}
L_0\bigl(\mathcal P_\alpha^{(s)}\bigr)
&=
\#\left\{
    \lambda\in\mathbb{F}_{q}
    \;\middle|\;
    \{(\lambda,v):v\in\mathbb{F}_{q}\}
    \cap\mathcal{P}_{\alpha}^{(s)}=\emptyset
\right\}+
\#\left\{
    (\beta,\gamma)\in\mathbb{F}_{q}^{2}
    \;\middle|\;
    \{(u,\beta u+\gamma):u\in\mathbb{F}_{q}\}
    \cap\mathcal{P}_{\alpha}^{(s)}=\emptyset
\right\}
\\[6pt]
&=
\sum_{\lambda\in\mathbb{F}_{q}}
\prod_{(u,v)\in\mathcal{P}_{\alpha}^{(s)}}
(u-\lambda)^{q-1}+
\sum_{\beta,\gamma\in\mathbb{F}_{q}}
\prod_{(u,v)\in\mathcal{P}_{\alpha}^{(s)}}
(v-\beta u-\gamma)^{q-1}.
\end{align*}
 since
 \begin{align*}
     \prod_{(u,v)\in\mathcal P_\alpha^{(s)}}
    (u-\lambda)^{q-1}=\begin{cases}
        1, & X=\lambda\ \text{avoids}\ \mathcal P_\alpha^{(s)} \\
        0, &\text{otherwise}
    \end{cases}
\quad \text{and} \quad 
    \prod_{(u,v)\in\mathcal P_\alpha^{(s)}}
    (v-\beta u-\gamma)^{q-1}=\begin{cases}
        1, &Y=\beta X+\gamma \ \text{avoids}\ \mathcal P_\alpha^{(s)}\\
        0, &\text{otherwise}
    \end{cases}
 \end{align*}
where each product is regarded as a $0$-$1$ indicator in $\mathbb{F}_q$,
and its value is identified with the corresponding integer $0$ or $1$
before summation.
 \qed

\begin{Rem}\label{full-field-codimension-three}
Suppose that $\{\alpha_1,\dots,\alpha_r\}=\mathbb F_q$,
$s\geq 2$, and $k=sq-3$.
Then
\[
\mathcal P_\alpha^{(s)}
=\{(x+y,xy)\mid x,y\in\mathbb F_q\}
\]
meets every affine line in $\mathbb F_q^2$.
Indeed, each vertical line $u=\lambda$ contains
$(\lambda,0)\in\mathcal P_\alpha^{(s)}$.
For a nonvertical line $v=\beta u+\gamma$, taking
$x=\beta+1$ and $y=\beta+\beta^2+\gamma$ gives
$xy=\beta(x+y)+\gamma$, so that $(x+y,xy)$ lies on the line.
Thus $L_0(\mathcal P_\alpha^{(s)})=0$, and the preceding
proposition yields
\[
\#\operatorname{DH}(\operatorname{HRS}_k(\alpha,V))
=(q-1)q^k.
\]
By Theorem~\ref{When Ev(f) deep hole}\textup{(1)} and
Lemma~\ref{iso}, the polynomials of degree $k$ already
give exactly $(q-1)q^k$ distinct NRT deep holes.
Thus the degree-$k$ polynomials account for all NRT deep holes
when $s\geq2$ and $k=sq-3$, for both odd and even $q$. This gives an independent verification of
Theorem~\ref{Deep holes of full-length HRS codes}
for $k=sq-3$ whenever $k\geq s$.
\end{Rem}

\section{The Cheng--Murray conjecture on deep holes
of full-length HRS codes}
In this section, we always consider $\{\alpha_{1},\dots,\alpha_{r}\}=\mathbb{F}_{q}$, so $r=q$ and $\pi_{\alpha}(x)=x^{q}-x$. We focus on $s\geq 2$. Together, Theorem~\ref{Deep holes of full-length HRS codes} and Proposition~\ref{Conjecture2} classify all NRT deep holes for $2\leq k\leq sq-1$ whenever $q$ is odd or $k\geq 3$. More precisely, we first  prove the degree-$k$ classification when $k\geq s$, and then use
row lengthening to obtain the explicit low-layer description when $k<s$.
\begin{Lem}\cite[Lemma 14]{Baj16}\label{sum set}
    Suppose that $d$ and $t$ are positive integers with $t\leq d-1$, and let $j\in\mathbb{Z}_{d}$. Then there is a $t$-subset $J=\{j_{1},\dots,j_{t}\}$ of $\mathbb{Z}_{d}$ for which $$j_{1}+j_{2}+\cdots+j_{t}=j.$$
\end{Lem}

\begin{Lem}\label{Products with bounded multiplicities}
    Let $q$ be a prime power, $r\geq 1$ and $h\in[r(q-1)-1]$. Then
    \begin{align*}
        \left\{\prod_{u\in\mathbb{F}_{q}^{*}}u^{\mu_{u}}\ \middle| \ \mu_{u}\in[0,r],\sum_{u\in\mathbb{F}_{q}^{*}}\mu_{u}=h\right\}=\mathbb{F}_{q}^{*}.
    \end{align*}
\end{Lem}
\textit{Proof}. Let $\mathbb{F}_{q}^{*}=\langle\gamma\rangle$, and define the group epimorphism $\varphi:\mathbb{Z}_{r(q-1)}\to\mathbb{F}_q^*$ by $\varphi(j)=\gamma^j$. For all $a\in\mathbb F_{q}^{*}$, and choose $b\in \varphi^{-1}(a)$. By Lemma~\ref{sum set}, there exists $B=\{b_{1},\dots,b_{h}\}\subseteq\mathbb{Z}_{r(q-1)}$ such that
$
\sum_{j=1}^{h}b_{j}=b.
$
Set
$
\mu_u=\bigl|B\cap\varphi^{-1}(u)\bigr|
$ for all $u\in\mathbb F_q^*$. Then  we have 
$0\leq \mu_{u}\leq \left|\varphi^{-1}(u)\right|=r$, $
\displaystyle\sum_{u\in\mathbb F_q^*}\mu_u=|B|=h.
$ and
$$
\prod_{u\in\mathbb F_q^*}u^{\mu_u}=
\prod_{u\in\mathbb F_q^*}\left(\prod_{b_{j}\in B \atop \varphi(b_{j})=u} u\right)
=\prod_{j=1}^{h}\varphi(b_{j})
=\varphi\left(\sum_{j=1}^{h}b_{j}\right)
=\varphi(b)
=a.
$$\qed

\begin{Def}\cite[p. 1924,1925,1929]{Ber10}
   \begin{enumerate}
       \item  Let $V$ be a vector space of dimension $\ell$ over a field $K$ of arbitrary characteristic. Let $\Delta=(\alpha_{1},\dots,\alpha_{n})$ be a sequence of elements spanning $V^{*}$, the dual space of $V$. 
    \begin{enumerate}
        \item The symmetric algebra $\operatorname{Sym}(V^{*})$ is defined as the $K$-algebra of polynomial functions on $V$, i.e. $\operatorname{Sym}(V^{*})=K[\alpha_{1},\dots,\alpha_{n}]$,
        and the $d$-th symmetric power of $V^*$ is defined as $\operatorname{Sym}^{d}(V^{*})=K[\alpha_{1},\dots,\alpha_{n}]_{d}.$
        \item The $K$-subspace $P(\Delta)$ of $\operatorname{Sym}(V^{*})$ is defined as $P(\Delta)=\operatorname{Span}_{K}\left\{\displaystyle\prod_{i\in S}\alpha_{i}\mid S\subseteq [n] \right\}.$
    \end{enumerate}
       \item For a subset $S\subseteq[n]$, the rank of $S$ is defined as $
\operatorname{rk}(S)
=
\dim_{\mathbb K}
\operatorname{span}_{\mathbb K}\{\alpha_i\mid i\in S\}.
$ A subset $S\subseteq [n]$ is called a flat if
$
\operatorname{rk}(S\cup\{i\})>\operatorname{rk}(S)
$ for all $i\in [n]\backslash S.
$
A cocircuit of $\Delta$ is a set $[n]\backslash S$, where
$S$ is a flat of rank $\ell-1$.
   \end{enumerate}
\end{Def}

\begin{Theo}[{\cite[Theorem~5.1]{Ber10}}]
\label{thm:berget-spanning}
There is a containment
$
\operatorname{Sym}^d(V^*)\subset P(\Delta)
$
if and only if $d$ is less than or equal to the size of
the smallest cocircuit of $\Delta$.
\end{Theo}

\begin{Lem}\label{Span of bounded divisors}
    Let $r\geq 1$ and $d\in[0,r(q-1)]$. Then $\operatorname{Span}_{\mathbb{F}_q}
    \left\{
        \prod_{a\in\mathbb{F}_q}(x-a)^{\mu_a}
        \;\middle|\;
        \mu_a\in[0,r],\
        \sum_{a\in\mathbb{F}_q}\mu_a=d
    \right\}=\mathbb{F}_{q}[x]_{\leq d}.$
\end{Lem}
\textit{Proof}.
The case $d=0$ is immediate.
Suppose that $r\geq 1$ and $d\in[r(q-1)]$, take $K=\mathbb F_q$, $V=\mathbb F_q^2$, $n=rq$
and $\Delta=(X-aY)_{(a,j)\in\mathbb{F}_{q}\times[r]}$. Then $\operatorname{Sym}(V^{*})=\mathbb{F}_{q}[X,Y]$, $\operatorname{Sym}^{d}(V^{*})=\left\{\displaystyle\sum_{i=0}^{d}c_{i}X^{d-i}Y^{i}\ \middle | \ c_{i}\in\mathbb{F}_{q} \right\}$ and $P(\Delta)=\operatorname{Span}_{\mathbb{F}_{q}}\left\{\displaystyle\prod_{a\in\mathbb{F}_{q}}(X-aY)^{\mu_{a}}\ \middle |\ \mu_{a}\in[0,r] \right\}.$ We show that the size of the smallest cocircuit of $\Delta$ is $r(q-1)$. Set
$
S_a=\{a\}\times\{1,\ldots,r\}
$ for all $a\in\mathbb{F}_{q}$
Since $X-aY$ and $X-bY$ are linearly independent whenever
$a\neq b$, the rank-one flats of $\Delta$ are precisely
the sets $S_a$, $a\in\mathbb F_q$. Consequently, the
cocircuits of $\Delta$ are exactly the sets $(\mathbb{F}_{q}\times[r])\backslash S_a$,
all of which have size $r(q-1)$.

It follows from Theorem~\ref{thm:berget-spanning} that
$
\mathbb F_q[X,Y]_d
=
\operatorname{Sym}^d(V^*)
\subseteq P(\Delta).
$
Hence
\begin{align*}
    \mathbb F_q[X,Y]_d=\operatorname{Sym}^d(V^*)=P(\Delta)\cap\operatorname{Sym}^d(V^*)=\operatorname{Span}_{\mathbb F_q}
\left\{
\prod_{a\in\mathbb F_q}(X-aY)^{\mu_a}
\;\middle|\;
\mu_a\in\{0,\ldots,r\},\
\sum_{a\in\mathbb F_q}\mu_a=d
\right\}.
\end{align*}
Applying the linear isomorphism
$
\mathbb F_q[X,Y]_d
\rightarrow \mathbb F_q[x]_{\leq d}$, defined by $
f(X,Y)\mapsto f(x,1)
$, to the preceding equality yields the desired conclusion. \qed

\begin{Lem}\label{A direct high-dimension range}
    Let $q$ be a prime power, $s\geq 1$  and $k\in[sq-1]$. If $k\geq s+q-3$, then 
    $$\operatorname{DH}(\operatorname{HRS}_{k}(\alpha,V))=\{\operatorname{Ev}_{\alpha,V}(f)\mid \deg(f)=k\}.$$
\end{Lem}
\textit{Proof}. The case $k=sq-1$ follows directly from the proof of Proposition~\ref{k=sr-1}. If $k\in[sq-2]$, let
$f(x)\in\mathbb F_q[x]_{\leq sq-1}$ with $\operatorname{Ev}_{\alpha,V}(f)\in
\operatorname{DH}(\operatorname{HRS}_k(\alpha,V))$. Note that $$\left\{\frac{\pi_{\alpha}^{s}}{(x-a)g(x)}\ \middle| \ a\in\mathbb{F}_{q},\text{ monic }g(x)\in\mathbb{F}_{q}[x]_{sq-k-2}\text{ with }g(x)\mid\pi_{\alpha}^{s-1} \right\}\subseteq\{M_{\boldsymbol{\mu}}\mid \boldsymbol{\mu}\in\mathcal{M}_{k}\}$$ and for all $a\in\mathbb{F}_{q}$, monic $g(x)\in\mathbb{F}_{q}[x]_{sq-k-2}$ with $g(x)\mid\pi_{\alpha}^{s-1}$, $$f(x)(x-a)g(x)\equiv\left(f(x)\bmod\frac{\pi_{\alpha}^{s}}{(x-a)g(x)}\right)(x-a)g(x)\bmod \pi_{\alpha}^{s}.$$
Hence, by Theorem~\ref{Deep holes}, we have
$
\deg\bigl(f(x)(x-a)g(x)\bmod\pi_{\alpha}^{s}\bigr)=sq-1$. We claim that
$
\deg(fg\bmod\pi_{\alpha}^{s})\leq sq-2.
$
Otherwise, $[x^{sq-1}](fg\bmod\pi_{\alpha}^{s})\neq0$. Set $
a_{0}=
\frac{
[x^{sq-1}](xfg\bmod\pi_{\alpha}^{s})
}{
[x^{sq-1}](fg\bmod\pi_{\alpha}^{s})
}
\in\mathbb F_q
$
and then
\begin{align*}
    [x^{sq-1}]
\bigl(f(x)(x-a_{0})g(x)\bmod\pi_{\alpha}^{s}\bigr)=[x^{sq-1}]
\bigl(xfg\bmod\pi_{\alpha}^{s}\bigr)-a_{0}[x^{sq-1}]
\bigl(fg\bmod\pi_{\alpha}^{s}\bigr)=0.
\end{align*}
which implies that $\deg\bigl(f(x)(x-a_{0})g(x)\bmod\pi_{\alpha}^{s}\bigr)\leq sq-2.$ This is a contradiction.
Since $k\geq s+q-3$, we have
$sq-k-2\leq(s-1)(q-1)$.
If $s\geq2$, by Lemma~\ref{Span of bounded divisors}, we have
$$\mathbb{F}_{q}[x]_{\leq sq-k-2}=\operatorname{Span}_{\mathbb{F}_{q}}\left\{g(x)\in\mathbb{F}_{q}[x]_{sq-k-2}: g(x)\mid \pi_{\alpha}^{s-1}\right\}.$$If $s=1$, then $k=q-2$, and the above equality
reduces to
$
\mathbb F_q[x]_{\leq0}
= \mathbb{F}_{q}
.$
This equality clearly holds.
Consequently, for all $i\in[0,sq-k-2]$, we have
$$
\deg(x^i f\bmod\pi_{\alpha}^{s})\leq sq-2.
$$
Suppose that $\deg f\geq k+1$, then
$sq-1-\deg f\in[0,sq-k-2]
$ and
$
\deg\bigl(
x^{sq-1-\deg f}f\bmod\pi_{\alpha}^{s}
\bigr)=sq-1,
$
This is a contradiction. Thus, $\deg f\leq k$, and Theorem~\ref{Deep holes}
forces $\deg f=k$.
\qed

\begin{Lem}\label{Rigidity on adding one Hasse layer}
    Suppose $q\geq 3$, $\ell\in[2,s]$ and $ \ell\leq k\leq (\ell-1)q-1$. Let $f_{\ell}\in\mathbb{F}_{q}[x]_{\leq \ell q-1}$ with $\operatorname{Ev}_{\alpha,V^{(\ell)}}(f_{\ell})\in\operatorname{DH}(\operatorname{HRS}_{k}(\alpha,V^{(\ell)})).$ If $\deg(f_{\ell}\bmod\pi_{\alpha}^{\ell-1})=k$, then $\deg(f_{\ell})=k$.
\end{Lem}

\textit{Proof}.
By assumption and the degree bounds, we may write $f_{\ell}\bmod\pi_{\alpha}^{\ell-1}=ax^k+g(x)$ and
$
f_\ell=ax^{k}+g(x)+\pi_\alpha^{\ell-1}(x)h(x)
$, where $ a\in\mathbb{F}_{q}^{*}$,           $g(x)\in\mathbb{F}_{q}[x]_{\leq k-1}$ and $h(x)\in\mathbb{F}_{q}[x]_{\leq q-1}$. We then prove that $h(\alpha_{i})=0$ for all $i\in[q]$.

Suppose that there exists $i_{0}\in[q]$ such that $h(\alpha_{i_{0}})\neq 0$.
Since
$
1\le k+1-\ell\le(\ell-1)(q-1)-1
$, by
Lemma~\ref{Products with bounded multiplicities}, there exists $(\mu_i)_{i\in[q]\backslash\{i_{0}\}}\in[0,\ell-1]^{q-1}$ such that
$
\displaystyle\sum_{i\ne i_{0}}\mu_i=k+1-\ell,
$
and
$
\prod_{i\ne i_{0}}(\alpha_{i_{0}}-\alpha_i)^{\mu_i}
=
\frac{(-1)^\ell h(\alpha_{i_{0}})}{a}\in\mathbb{F}_{q}^{*}.
$
Set
$
M(x)=(x-\alpha_{i_{0}})^\ell
\prod_{i\ne i_{0}}(x-\alpha_i)^{\mu_i}.
$
Then $M(x)\in \{M_{\boldsymbol{\boldsymbol{\mu}}}\mid \boldsymbol{\mu}\in\mathcal{M}_{k}\}$.
Since $\pi_\alpha'(\alpha_i)=-1$ for all $i\in[q]$, we have
\begin{align*}
    \pi_\alpha^{\ell-1}h
&=(x-\alpha_{i_{0}})^{\ell-1}\left(\left(\frac{\pi_{\alpha}(x)}{x-\alpha_{i_{0}}}\right)^{\ell-1}h(x)\right)\equiv (x-\alpha_{i_{0}})^{\ell-1}(-1)^{\ell-1}h(\alpha_{i_{0}}) \bmod{(x-\alpha_{i_{0}})^\ell}\\
&=-a(x-\alpha_{i_{0}})^{\ell-1}\prod_{i\neq i_{0}}(\alpha_{i_{0}}-\alpha_{i}
)^{\mu_{i}}\equiv
-a\frac{M(x)}{x-\alpha_{i_{0}}}
\bmod{(x-\alpha_{i_{0}})^\ell}.
\end{align*}
Clearly, we also have $\pi_{\alpha}^{\ell-1}h\equiv-a\frac{M(x)}{x-\alpha_{i_{0}}}
\bmod{(x-\alpha_i)^{\mu_{i}}}$ for all $i\neq i_{0}$. It follows that
$
\pi_\alpha^{\ell-1}h
\equiv
-a\frac{M(x)}{x-\alpha_{i_{0}}}
\bmod M.
$
Consequently,
$$
f_\ell
\equiv
ax^{k}+g(x)-a\frac{M(x)}{x-\alpha_{i_{0}}}
\bmod M,
$$
where  $\deg(
ax^{k}+g(x)-a\frac{M(x)}{x-\alpha_{i_{0}}})\leq k-1$.
Hence $
[x^k](f_\ell\bmod M)=0,
$ which contradicts the assumption that $\operatorname{Ev}_{\alpha,V^{(\ell)}}(f_{\ell})\in\operatorname{DH}(\operatorname{HRS}_{k}(\alpha,V^{(\ell)})).$ Thus $h(\alpha_i)=0$ for all $i\in[q]$. This implies that $h=0$ and $f_{\ell}=ax^{k}+g(x).$
\qed

\begin{Theo}\label{All characteristics, high dimensions}
    Let $q$ be a prime power and $2\leq s\leq k\leq sq-1$. If $k\geq q-1$, then $$\operatorname{DH}(\operatorname{HRS}_{k}(\alpha,V))=\{\operatorname{Ev}_{\alpha,V}(f)\mid \deg (f)=k\}.$$
\end{Theo}
\textit{Proof.}  We prove this by induction on $s.$
Clearly, the theorem holds for $s=2$ by
Lemma~\ref{A direct high-dimension range}. Assume that the theorem holds for $s-1$. If  $k\ge s+q-3$, the theorem follows from
Lemma~\ref{A direct high-dimension range}.
Hence it remains to consider the case $2\leq s\leq  k<s+q-3$.
This implies $q\ge4$ and $s-1<s\le k\le s+q-4\le (s-1)q-1.$
By the induction hypothesis, we have $$\operatorname{DH}(\operatorname{HRS}_{k}(\alpha,V^{(s-1)}))=\{\operatorname{Ev}_{\alpha,V^{(s-1)}}(f)\mid \deg (f)=k\}.$$
Now, suppose $U=\operatorname{Ev}_{\alpha,V}(f)\in\operatorname{DH}(\operatorname{HRS}_{k}(\alpha,V))$,
where $\deg f\leq sq-1$. Set $f_{s-1}=f\bmod\pi_\alpha^{s-1}.$
By Theorem~\ref{row truncation} and Theorem~\ref{NRT covering radius},
\[
U^{(s-1)}
=\operatorname{Ev}_{\alpha,V^{(s-1)}}(f_{s-1})
\in\operatorname{DH}(\operatorname{HRS}_{k}(\alpha,V^{(s-1)})),
\]
which implies $\deg (f_{s-1})=k.$
It follows from 
Lemma~\ref{Rigidity on adding one Hasse layer} that $\deg(f)=k.$
\qed

\begin{Lem}\label{Reduction to a fixed RS dimension}
    Let $q$ be a prime power, $s\geq 2$ and
    $2\leq d\leq k\leq q-1$. 
    Suppose that
    \[
    \operatorname{DH}(\operatorname{HRS}_{d}(\alpha,V^{(1)}))
    =
    \{\operatorname{Ev}_{\alpha,V^{(1)}}(f)\mid \deg(f)=d\}.
    \]
    Then $\deg(f_{U}\bmod\pi_{\alpha})=k$ for all $U\in\operatorname{DH}(\operatorname{HRS}_{k}(\alpha,V)).$
\end{Lem}

\textit{Proof}.
Suppose that $U\in\operatorname{DH}(\operatorname{HRS}_{k}(\alpha,V))$.
Choose a monic divisor $B$ of $\pi_\alpha$ of degree $k-d$,
and write $f_U=g+Bh$, where $g=0$
or $\deg(g)<k-d$. Set $h_0=h\bmod\pi_\alpha\in\mathbb{F}_{q}[x]_{\leq q-1}$. Then we show that $\deg(h_{0})=d$.

For every monic divisor $M$ of $\pi_\alpha$ of degree $d+1$,
since $BM\mid\pi_\alpha^s$ and $\deg(BM)=k+1$,
we have $BM=M_{\boldsymbol{\mu}}$ for some
$\boldsymbol{\mu}\in\mathcal M_k$.
 It follows from Theorem~\ref{Deep holes} that
$\deg(f_U\bmod BM)=k$.
Moreover,
\[
f_U\equiv g+B(h\bmod M)\pmod{BM},
\qquad
\deg\bigl(g+B(h\bmod M)\bigr)\leq k.
\]
Thus
\[
f_U\bmod BM=g+B(h\bmod M),\quad \deg(h\bmod M)=d.
\]
On the other hand, since $M\mid\pi_\alpha$, we have
$h_0\bmod M=h\bmod M$, which implies that $\deg(h_{0}\bmod M)=d$.
By Theorem~\ref{Deep holes} and the hypothesis, we obtain $\operatorname{Ev}_{\alpha,V^{(1)}}(h_0)
\in\operatorname{DH}(\operatorname{HRS}_{d}(\alpha,V^{(1)}))$ and $\deg(h_0)=d$.

Finally, $f_U\equiv g+Bh_0\bmod\pi_\alpha$ and
$\deg(g+Bh_0)=k\leq q-1$. Therefore $f_U\bmod\pi_\alpha=g+Bh_0$,
which gives $\deg(f_U\bmod\pi_\alpha)=k$.
\qed

\begin{Theo}\label{Deep holes of full-length HRS codes}
    Let $q$ be a prime power.
    Suppose that $2\leq s\leq k\leq sq-1$.
    If $q$ is odd or $k\geq3$, then
    \[
    \operatorname{DH}(\operatorname{HRS}_k(\alpha,V))
    =
    \{\operatorname{Ev}_{\alpha,V}(f)\mid\deg(f)=k\}.
    \]
\end{Theo}

\textit{Proof}.
The case $k\geq q-1$ follows from
Theorem~\ref{All characteristics, high dimensions}.
Assume that $k\leq q-2$. Let
$U\in\operatorname{DH}(\operatorname{HRS}_k(\alpha,V))$,
and set $f_\ell=f_U\bmod\pi_\alpha^\ell$ for $\ell\in[s]$.

Take $d=2$ if $q$ is odd and $d=3$ if $q$ is even.
Then $2\leq d\leq k$. In the even case, $3\leq k\leq q-2$
implies $q\geq8$, so $d\leq\lfloor q/2\rfloor$.
By \cite[Proposition~2]{Kai17} for odd $q$ and
\cite[Theorem~1.2]{F26} for even $q$, the deep holes of the
full-length Reed--Solomon code of dimension $d$ are precisely
the words whose canonical interpolation polynomials have degree $d$. Therefore, we have $$\operatorname{DH}(\operatorname{HRS}_{d}(\alpha,V^{(1)}))
    =
    \{\operatorname{Ev}_{\alpha,V^{(1)}}(f)\mid \deg(f)=d\}$$
    where $d=2$ if $q$ is odd and $d=3$ if $q$ is even.
It follows from Lemma~\ref{Reduction to a fixed RS dimension}
that
\[
\deg(f_U\bmod\pi_\alpha)=k.
\]
Note that $\operatorname{HRS}_k(\alpha,V^{(\ell)})=\operatorname{HRS}_k(\alpha,V)^{(\ell)}$ for all $\ell \in[s].$ 
By Theorem~\ref{row truncation} and Theorem~\ref{NRT covering radius}, we obtain
\[
U^{(\ell)}
=\operatorname{Ev}_{\alpha,V^{(\ell)}}(f_\ell)
\in\operatorname{DH}(\operatorname{HRS}_k(\alpha,V^{(\ell)})).
\]
For every $\ell\in[2,s]$, we have $f_\ell\bmod\pi_\alpha^{\ell-1}=f_{\ell-1}$ and $\ell\leq s\leq k\leq q-2\leq(\ell-1)q-1.$
Since $q\geq4$ and $\deg(f_1)=k$, successive applications of
Lemma~\ref{Rigidity on adding one Hasse layer} for $\ell=2,\ldots,s$
yield $\deg(f_s)=k$.
\qed

   

\begin{Prop}\label{Conjecture2}
  Let $q$ be a prime power.
    Suppose that $2\le k\le s-1$.
    If $q$ is odd or $k\geq3$, then \begin{align*}
            \operatorname{DH}(\operatorname{HRS}_{k}(\alpha,V))=\left\{\operatorname{Ev}_{\alpha,V}(ax^{k}+g(x)+\pi_{\alpha}(x)^{k}h(x))\ \middle| \ \begin{array}{c}
       a\in\mathbb{F}_{q}^{*},  g(x)\in\mathbb{F}_{q}[x]_{\leq k-1},h(x)\in\mathbb{F}_{q}[x]_{\leq(s-k)r-1}   \\
     \text{with}\   h(\alpha_{j})\neq -a\eta_{j}^{-k}\ \text{for} \ j\in[r]
    \end{array}\right\}.
        \end{align*}
\end{Prop}

{\it Proof.} It follows from Theorem~\ref{Deep holes of full-length HRS codes} that $$\operatorname{DH}(\operatorname{HRS}_{k}(\alpha,V^{(k)}))=\{\operatorname{Ev}_{\alpha,V^{(k)}}(f)\mid \deg(f)=k\}.$$
Since $\partial^{(j)}g(x)=0$ for $g(x)\in\mathbb{F}_{q}[x]_{\leq k-1}$ and $j=k,\dots,s-1$, the code $\operatorname{HRS}_{k}(\alpha,V)$ is precisely the $\mathbf{0}^{(s-k)\times r}$-row lengthening code of $\operatorname{HRS}_{k}(\alpha,V^{(k)})$. Then, by Theorem~\ref{deep hole of extend code}, we have
\begin{align*}
\operatorname{DH}\bigl(\operatorname{HRS}_{k}(\alpha,V)\bigr)
=
\left\{
    \begin{pmatrix}
        \operatorname{Ev}_{\alpha,V^{(k)}}(f)\\
        Z
    \end{pmatrix}
    \;\middle|\;
    \begin{aligned}
        &f\in\mathbb{F}_{q}[x]_{k},\quad
         Z\in\mathbb{F}_{q}^{(s-k)\times r},
        \\[3pt]
        &\Delta_{s-k}\bigl(
            \operatorname{Ev}_{\alpha,V^{(k)}}(f),Z;
            \operatorname{Ev}_{\alpha,V^{(k)}}(g)
        \bigr)\leq
        \operatorname{wt}_{\mathrm{RT}}\bigl(
            \operatorname{Ev}_{\alpha,V^{(k)}}(f-g)
        \bigr)-(kr-k),
        \\[3pt]
        &\text{for all }g\in\mathbb{F}_{q}[x]_{\leq k-1}
    \end{aligned}
\right\}.
\end{align*}
We claim that the first row of $Z$ lies in $(\mathbb{F}_{q}^{*})^{1\times r}$ for all $\begin{pmatrix}
    \operatorname{Ev}_{\alpha,V^{(k)}}(f)\\
  Z
\end{pmatrix}\in\operatorname{DH}(\operatorname{HRS}_{k}(\alpha,V))$. Otherwise, there exists $j_{0}\in [r]$ such that $z_{1,j_{0}}=0.$ This implies $wt_{\operatorname{RT}}(Z_{j_{0}})\leq s-k-1.$ Now we may write $f(x)=ax^{k}+g(x)$ and define $g_{0}(x)=ax^{k}+g(x)-a(x-\alpha_{j_{0}})^{k}\in\mathbb{F}_{q}[x]_{\leq k-1}$ for some $a\in\mathbb{F}_{q}^{*}$ and $g(x)\in\mathbb{F}_{q}[x]_{\leq k-1}$. Then  $\operatorname{Ev}_{\alpha,V^{(k)}}(f)_{j_{0}}=\operatorname{Ev}_{\alpha,V^{(k)}}(g_{0})_{j_{0}}$ and \begin{align*}
    1&\leq s-k-wt_{\operatorname{RT}}(Z_{j_{0}})\leq \Delta_{s-k}(\operatorname{Ev}_{\alpha,V^{(k)}}(f),Z;\operatorname{Ev}_{\alpha,V^{(k)}}(g_{0}))\\
      &\leq  wt_{\operatorname{RT}}(\operatorname{Ev}_{\alpha,V^{(k)}}(f-g_{0}))-(kr-k)= (kr-k)-(kr-k) =0
\end{align*}
This is a contradiction. So our claim follows.
We define $F(x):=f_{\begin{pmatrix}
    \operatorname{Ev}_{\alpha,V^{(k)}}(f)\\
    Z
\end{pmatrix}}$. By Proposition~\ref{adic expansion}, there exist $F_{i}(x)\in\mathbb{F}_{q}[x]_{<r}$ for  $i\in[0,s-1]$ such that 
$$F(x)=F_{0}(x)+F_{1}(x)\pi_{\alpha}(x)+\cdots+F_{s-1}(x)\pi_{\alpha}(x)^{s-1}.$$
Since $\partial^{(i)}(F_{k}\pi_{\alpha}^{k}+\cdots+F_{s-1}\pi_{\alpha}^{s-1})(\alpha_{j})=0$ for all  $i\in[0,k-1]$ and $j\in[r]$, the first $k$ Hasse layers of $F$ are determined by $$R(x):=F_{0}(x)+F_{1}(x)\pi_{\alpha}(x)+\cdots+F_{k-1}(x)\pi_{\alpha}(x)^{k-1}.$$ By the definition of $F$, these first $k$ layers coincide with those of $f$. Since both $R$ and $f$ have degree $<kr$, the uniqueness of Hermite interpolation modulo $\pi_{\alpha}^{k}$ implies that $R=f$. Consequently, $$F=f+\pi_{\alpha}^{k}h,$$ for some $h(x)\in\mathbb{F}_{q}[x]_{\leq (s-k)r-1}$. This implies that$$\partial^{(k)}F(\alpha_{j})=\partial^{(k)}f(\alpha_{j})+\partial^{(k)}(\pi_{\alpha}^{k}h)(\alpha_{j})=a+\eta_{j}^{k}h(\alpha_{j})\quad \text{for } j\in[r]$$ and that $$0\neq z_{1,j}=v_{k+1,j}\partial^{(k)}F(\alpha_{j})=v_{k+1,j}(a+\eta_{j}^{k}h(\alpha_{j}))\quad \text{for } j\in[r].$$ This is equivalent to $h(\alpha_{j})\neq -a\eta_{j}^{-k}$ since $v_{k+1,j}\in\mathbb{F}_{q}^{*}.$ Thus, combining Theorem~\ref{When Ev(f) deep hole}\textup{(3)} with the above characterization, we have \begin{align*}
            \operatorname{DH}(\operatorname{HRS}_{k}(\alpha,V))=\left\{\operatorname{Ev}_{\alpha,V}(ax^{k}+g(x)+\pi_{\alpha}(x)^{k}h(x))\ \middle| \ \begin{array}{c}
       a\in\mathbb{F}_{q}^{*},  g(x)\in\mathbb{F}_{q}[x]_{\leq k-1}, h(x)\in\mathbb{F}_{q}[x]_{\leq (s-k)r-1}   \\
       \text{with} \  h(\alpha_{j})\neq -a\eta_{j}^{-k}\ \text{for} \ j\in[r]
    \end{array}\right\}.
        \end{align*}\qed

\section*{Declaration of competing interests}

The authors declare that they have no known competing financial
interests or personal relationships that could have appeared to
influence the work reported in this paper.

\section*{Acknowledgments}
The research was supported by the National Natural Science Foundation of China under the Grant No. 12441105, Interdisciplinary Project of Capital Normal University (Grant No. 2026JCYB02) and Stability Program of National Key Laboratory of Security Communication (WD202609). The authors used ChatGPT-6.1 Sol Ultra for language editing and literature searches. All AI-assisted content and references were reviewed and verified by the authors, who take full responsibility for this work.

\section*{Data availability}
The results reported in this article do not rely on computational data.


\end{document}